\documentclass[english]{article}

\usepackage{amsmath,amsfonts,amsthm}
\usepackage{graphicx}
\usepackage{epstopdf}

\usepackage{booktabs}

\usepackage{tikz,tkz-tab}
\usepackage{hyperref}
\usepackage[overload]{empheq} 
\usepackage{color}
\usepackage{float}
\def\cN{{\cal N}}   
\newsavebox\myuglybox
\def\mye#1{                 
  \savebox{\myuglybox}{$#1$}
  \resizebox{!}{\ht\myuglybox}{$\mathrm{e}$}^{\usebox{\myuglybox}}} 

\newtheorem{thm}{Theorem}

\begin{document}
\title{Optimal Control Approach for Implementation of Sterile Insect Techniques}
\author{Pierre-Alexandre Bliman$^{1}$, Daiver Cardona-Salgado$^2$, Yves Dumont$^{3,4,5}$, Olga Vasilieva$^6$\footnote{Corresponding author: olga.vasilieva@correounivalle.edu.co}\\
$^1$ \small Sorbonne Universit\'e, Universit\'e Paris-Diderot SPC, \color{black} Inria, CNRS\\ \small Laboratoire Jacques-Louis Lions, \'equipe Mamba, \color{black}  Paris, France\\
$^2$ \small Universidad Aut\'{o}noma de Occidente, Cali, Colombia \\
$^3$ \small CIRAD, Umr AMAP, Pretoria, South Africa \\
$^4$ \small AMAP, Univ Montpellier, CIRAD, CNRS, INRA, IRD, Montpellier, France, \\
$^5$ \small University of Pretoria, Department of Mathematics and Applied Mathematics, South Africa\\ 
$^6$ \small Universidad del Valle, Cali, Colombia
}
\date{}
\maketitle

\begin{abstract}
Vector or pest control is essential to reduce the risk of vector-borne diseases or crop losses. Among the available biological control tools, the Sterile Insect Technique (SIT) is one of the most promising. However, SIT-control campaigns must be carefully planned in advance in order to render desirable outcomes. In this paper, we design SIT-control intervention programs that can avoid the real-time monitoring of the wild population and require to mass-rear a minimal overall number of sterile insects, in order to induce a local elimination of the wild population in the shortest time. Continuous-time release programs are obtained by applying an optimal control approach, and then laying the groundwork of more practical SIT-control programs consisting of periodic impulsive releases.  
\end{abstract}

\tableofcontents

\newpage

\section{Introduction}
The presence and abundance of mosquito species is strongly correlated with progressive dissemination and persistence of various vector-borne diseases in human populations, and thus constitutes a major issue to health-care authorities in many countries. In the absence of effective vaccine against malaria (transmitted by \emph{Anopheles} mosquito species), encephalitis and West Nile virus (transmitted by \emph{Culex} species), as well as dengue fever, Chikungunya and Zika viruses (spread by \emph{Aedes} species), the control of these infectious diseases is primarily centered on suppression of the local mosquito populations through the use of chemical substances (larvicides, insecticides).

However, due to the negative environmental impact of chemical substancies on non-target species and the detected resistance developed by mosquitoes to insecticides and larvicides, other methods have been proposed and evaluated. One of them is the so-called Sterile Insect Technique (SIT) that relies on massive releases of male insects sterilized by ionizing radiations or using incompatible strains of \emph{Wolbachia} symbiont. After mating with sterilized or \emph{Wolbachia}-carrying males, female mosquitoes duly lay eggs that never hatch into larvae, and the latter induces a progressive suppression of the target wild population (see \cite{Dyck2006} regarding the detailed description of SIT and its applications).

Similarly, pest, like fruit flies, are responsible of considerable losses in crops, and orchards, around the world. For instance, \textit{Bactrocera dorsalis}, a highly invasive species, was found in South Africa in 2010 \cite{Manrakhan2015}. It was first monitored in R\'eunion island in 2017\footnote{See the report given in  http://daaf.reunion.agriculture.gouv.fr/Detection-d-une-nouvelle-mouche}, and was also detected for the first time in Europe, southern Italy, in 2018. Unfortunately, to control such a pest, very few tools, apart from insecticides, exist.

SIT-control interventions have been modeled and studied theoretically and numerically in a large number of papers using temporal and spatiotemporal models with discrete, continuous or hybrid time scaling while applying open-loop, closed-loop or mixed modeling approaches to derive continuous-time and/or periodic impulsive release programs (see, e.g.,  \cite{Anguelov2012,Bliman2019,Dumont2012,Huang2017,Li2015,Strugarek2019} and references therein). However, the proposed programs for SIT-control required either to go on with excessively abundant releases for a considerable time or to monitor the current sizes of wild populations in order to reach an eventual elimination. Several scholars have also applied the dynamic optimization approach seeking to minimize the overall number of sterile insects to be released during the SIT-control campaign while anticipating a certain fixed duration of the control intervention \cite{Fister2013,Multerer2019}.

The principal  goal of this paper is to design the SIT-control intervention programs that exhibit a better tradeoff between different objectives set by the decision-makers. The primary objective of SIT-control campaigns is to ensure the elimination or drastic reduction in the number of wild insects within some target locality in a finite time. Other (secondary) objectives are related to reducing the underlying costs of SIT-control interventions, which are usually expressed by key indicators such as the overall duration of the SIT-control campaign, the number of releases to be performed, and the number of sterilized males to be mass-reared for releases. Some SIT-control programs require for periodic assessments of the wild population sizes in order to determine the sizes of subsequent releases. Such measurements constitute extra costs of the SIT-control programs and must be accounted for. Therefore, it is essential to identify the better tradeoff between different secondary objectives and to design a SIT-control program that requires lower overall costs. In other words, we are seeking to design the control intervention programs that require not only a moderate control effort (expressed through the cumulative number of sterile insects to be released), but also induce local elimination of wild insects in a reasonable time and without compelling for real-time assessments of the wild population sizes.

For that purpose, we make use of an entomological model initially proposed and thoroughly analyzed in \cite{Bliman2019} (its key features are given in Section ~\ref{sec-model}) and then formulate an optimal control problem with free terminal time and a multi-criteria objective functional (Section~\ref{sec-oc}). This problem is reduced, by applying the traditional technique \cite{Lenhart2007}, to an optimality system with a time-optimality condition that is solved numerically to obtain the optimal release programs in continuous time while considering different priorities in the decision-making. Furthermore, we adopt the continuous-time optimal release programs for practical implementation and propose ``almost optimal'' (further referred to as \emph{suboptimal}) SIT-control programs consisting of periodic impulsive releases of sterile insects. For the numerical simulations we consider the case of \textit{Aedes spp} mosquito population, like in \cite{Bliman2019}. The results of all numerical simulations together with an underlying discussion that highlights the performance of the purely open-loop suboptimal release programs as well as their comparison with mixed open/closed-loop impulsive release programs \cite{Bliman2019} are presented in Section ~\ref{sec-num}, and Section~\ref{sec-con} provides some final remarks and conclusions.

\section{A sex-structured entomological model involving sterile males}
\label{sec-model}

Let $M(t)$ and $F(t)$ be the current population sizes of males and females present in some target locality at the time $t \geq 0.$ Without loss of generality, we assume in the sequel that time is measured in days. Suppose that sterile males are being released in this locality at the per-day rate $u(t) \geq 0$ and denote their population size  $M_S(t)$  at each $t \geq 0$. The release rate can be modeled by a piecewise continuous function $u(\cdot) \in PC$ defined for all $t \geq 0$. This function represents the decision or control variable and expresses the number of sterile males $M_S$ introduced at the day $t$ in the target locality. In other words, $u(t), t \geq 0$ defines an external control action over the natural population dynamics of the wild subpopulations, $M$ and $F$.

Using the modeling framework proposed by Bliman \emph{et al} \cite{Bliman2019}, the population dynamics of wild insects ($M$ and $F$) and their interaction with sterile males ($M_S$) can be described by the following three-dimensional ODE system:

\begin{subequations}
\label{SITsys}
\begin{align}[left = \empheqlbrace\,]
\label{SITsys-a}
\dot M &=r\rho\dfrac{FM}{M+\gamma M_S} \mye{-\beta(M+F)}-\mu_M M, & & \hspace{-1mm} M(0) =M^0 > 0\\
\label{SITsys-b}
\dot F &=(1-r)\rho\dfrac{FM}{M+\gamma M_S} \mye{-\beta(M+F)}-\mu_F F, & & F(0) =F^0 > 0\\
\label{SITsys-c}
\dot M_S &=u(t)-\mu_S M_S, & & \hspace{-3mm} M_S(0) = M_S^0 \geq 0.
\end{align}
\end{subequations}
All the parameters of the model \eqref{SITsys} are positive, and their detailed descriptions are provided in Table~\ref{tab1}. The model assumes that all females $F$ are equally
able to mate, with either wild males $M$ or sterile males $M_S$. However, viable offspring is produced by wild females only after their successful mating with wild males. Therefore, the recruitment rate of wild insects (cf. positive terms in Eqs. \eqref{SITsys-a}-\eqref{SITsys-b}) depends on the probability $\dfrac{M}{M+ \gamma M_S} $ of effective contacts (or successful matings) between wild males and females. It is worth noting that sterilization of mass-reared male mosquitoes may affect their mating competitiveness \cite{Oliva2012}, and the latter is expressed in the model by the coefficient $0 < \gamma \leq 1$ denoting the relative mating efficiency or reproductive fitness of sterilized males (compared to the wild ones).

Eqs. \eqref{SITsys-a}-\eqref{SITsys-b} state that new individuals are recruited as either males or females in accordance with a primary sex ratio $r \div (1-r), r \in (0,1) $ in viable offsprings that survive to  adulthood. Here  $\rho$  stands for the mean number of eggs deposited by a single female in average per day, while direct and/or indirect competition effect at different stages (larvae, pupae, adults) are included through the parameter $\beta$. Thus, lower values of $\beta$ imply that a larger fraction of eggs may survive to adulthood, whereas its higher values express stronger competition and fewer breeding sites.

\begin{table}[h]
\centering
\begin{tabular}{|c|l|c|c|}
\hline
\textbf{Parameter} & \textbf{Description} & \textbf{Unit} & \textbf{Value (\textit{Aedes spp})} \\ \hline
$r$      & Primary sex ratio                                        & --         & $0.5$ \\
$\rho$   &  Average number of eggs deposited per female per day     & day$^{-1}$ & $4.55$ \\
$\beta$  &  Characteristic of the competition effect per individual & --         & $3.57 \times 10^{-4}$  \\
$\gamma$ &  Coefficient of mating efficiency of sterile males       & --         & $1$    \\
$\mu_M$  &  Natural mortality rate for male mosquitoes    & day$^{-1}$ & $0.04$ \\
$\mu_F$  &  Natural mortality rate for female mosquitoes  & day$^{-1}$ & $0.03$ \\
$\mu_S$  &  Natural mortality rate for sterile males      & day$^{-1}$ & $0.04$ \\
\hline
\end{tabular}
\caption{\textit{Aedes spp} mosquito parameters of the model \eqref{SITsys} with values borrowed from \cite{Bliman2019}}
\label{tab1}
\end{table}

Parameters $\mu_M, \mu_F,$ and $\mu_S$ in \eqref{SITsys} represent, respectively, the natural sex-specific mortality rates for wild males, wild females, and sterile males. Scientific evidence attests that mating females have higher average longevity than mating males \cite{Liles1965}, whereas sterilization of males may reduce their average lifespan \cite{Oliva2012}. Therefore, we suppose throughout the paper that:
\begin{equation}
\label{death-rates}
\mu_S \geq \mu_M \geq \mu_F.
\end{equation}

Finally, Eq. \eqref{SITsys-c} states that sterile males $M_S$ are ``exogenously'' released at a time-dependent instantaneous rate $u(t)$. Here we assume that $u: \mathbb{R}_+ \mapsto [0, u_{\max}]$ is a piecewise continuous function, and $u_{\max} >0$ is the maximum capacity of releases. For a known $u(t)$, Eq. \eqref{SITsys-c} can be explicitly solved:
\begin{equation}
\label{M_S-u-sol}
M_S \big( t; u(\cdot) \big)=\mye{-\mu_S t}  \left[ M_S^0 + \int \limits_0^t \mye{\mu_S \xi} u(\xi) d\xi \right]
\end{equation}

According to \eqref{M_S-u-sol}, $M_S \big( t;u(\cdot) \big)$ remains nonnegative and bounded from above whenever $u(t)$ is nonnegative and bounded from above. Additionally, the right-hand sides of Eqn. \eqref{SITsys-a}-\eqref{SITsys-b} are decreasing in $M_S(t)$. On the other hand, Bliman \emph{et al} \cite{Bliman2019} have shown that, in absence of sterile males (that is, with $M_S(t) = 0$), the sex-structured entomological model\footnote{This model is obtained from Eqs. \eqref{SITsys-a}-\eqref{SITsys-b} when $M_S(t)= 0$ for al $t \geq 0$.} has bounded trajectories that belong to an \emph{absorbing} set
\begin{equation}
\label{d-set}
\mathcal{D} = \Big\{ \big( M(t), F(t) \big): \: 0 \leq M(t) \leq C, \: 0 \leq F(t) \leq C, \:\: \forall \: t \geq 0 \Big\}
\end{equation}
for some $C>0$.

Since the right-hand sides of the dynamical system in \eqref{SITsys} are Lipschitz-continuous, the initial-value problem \eqref{SITsys} has a unique solution for all piecewise continuous and bounded $u(t) \in [0,u_{\max}], t \geq 0$.

In the context of this model, the upper bound of $u(t)$ (expressed by the maximum release capacity $u_{\max}$) plays an essential role, and its value must be defined in accordance with the model's parameters (presented in Table ~\ref{tab1}). A recent study performed by Bliman \emph{et al} \cite{Bliman2019} can shed some light on an adequate choice of $u_{\max}$.

In particular, the authors of \cite{Bliman2019} had thoroughly analyzed the case when $u(t) = \Lambda, \ \Lambda=\text{const} >0$ for all $t\geq 0$. Under such settings, Eq. \eqref{SITsys-c} becomes completely decoupled from the system \eqref{SITsys} and its steady-state value $M_S^{\sharp} = \dfrac{\Lambda}{\mu_S}$ can be replaced in Eqs.  \eqref{SITsys-a}-\eqref{SITsys-b} thus leading to the reduced ODE system
\begin{subequations}
\label{SIT-reduced}
\begin{align}[left = \empheqlbrace\,]
\label{SIT-reduced-a}
\dot M &=r\rho\dfrac{FM}{M+\gamma M_S^{\sharp}} \mye{-\beta(M+F)}-\mu_M M \\
\label{SIT-reduced-b}
\dot F &=(1-r)\rho\dfrac{FM}{M+\gamma M_S^{\sharp}} \mye{-\beta(M+F)}-\mu_F F
\end{align}
\end{subequations}
for naturally persistent wild populations. Let us recall that pest or mosquito populations are referred to as ``naturally persistent'' when one individual (male or female) produces, on average, more than one individual. It fact, Bliman \emph{et al} \cite{Bliman2019} introduced two positive constants, $\cN_F$ and $\cN_F$, that represent the so-called {\em basic offspring numbers} related to the wild female and male populations, respectively. In the sequel, we will be dealing with naturally persistent insect populations. Therefore, we suppose that for the parameters of the model (defined in Table~\ref{tab1}) it holds that
\begin{equation}
\label{bon}
\cN_F:=\dfrac{(1-r)\rho}{\mu_F } > 1, \qquad \cN_M:=\dfrac{r\rho}{\mu_M } > 1.
\end{equation}

For the reduced system \eqref{SIT-reduced}, the following fundamental results has been established in \cite{Bliman2019}.

\begin{thm}[Existence of positive equilibria for the reduced SIT system \eqref{SIT-reduced}, \cite{Bliman2019}]
\label{th1}
Assume that $\cN_F >1$. Then there exists  $\Lambda^{\text{crit}}>0$ such that system \eqref{SIT-reduced} admits
\begin{itemize}
\item
two positive distinct  equilibria if $0 < \Lambda < \Lambda^{\text{crit}}$,
\item
one positive equilibrium if $\Lambda = \Lambda^{\text{crit}}$,
\item
no positive equilibrium if $\Lambda > \Lambda^{\text{crit}}$.
\end{itemize}
The value of $\Lambda^{\text{crit}}$ is uniquely determined by the formula
\begin{equation}
\label{lambda_crit}
\Lambda^{\text{crit}} := 2 \: \frac{\mu_S}{\beta\gamma} \: \frac{\phi^{\text{crit}}(\cN_F)}{1+\frac{\cN_F}{\cN_M}},
\end{equation}
where $\phi^{\text{crit}}:=\phi^{\text{crit}}(\cN_F)$ is the unique positive solution to the transcendental equation
\begin{equation}
\label{phi}
1+ \phi\left( 1+ \sqrt{1+\frac{2}{\phi}} \right) = \cN_F \exp\left[ -\frac{2}{1+ \sqrt{1+\dfrac{2}{\phi}}} \right].
\end{equation}
\end{thm}

\begin{thm}[Stability of the insect-free equilibrium, \cite{Bliman2019}]
\label{th2}
If the reduced system \eqref{SIT-reduced} admits no positive equilibrium (that is, if $\Lambda>\Lambda^{\text{crit}}$), then the trivial equilibrium $ (0,0) \in \mathbb{R}^2$ of \eqref{SIT-reduced} is globally exponentially stable.
\end{thm}
Formal proofs of Theorems ~\ref{th1} and ~\ref{th2} are given in \cite{Bliman2019}. These theorems jointly provide a threshold value $\Lambda_{\text{crit}} > 0$ of the constant release rate such that for all \emph{constant} release rates $ u(t) = \Lambda > \Lambda_{\text{crit}}, t \geq 0$ the trajectories of the SIT system \eqref{SITsys} engendered by any nonnegative set of initial conditions $\big( M^0, F^0, M^0_S \big)$ converge globally exponentially toward its unique equilibrium $E^{\sharp}_0= \big( 0,0,  M_S^{\sharp} \big)$ free of wild insects, and the latter ultimately ensures an eventual elimination of the wild population.

Therefore, it is reasonable to take $u_{\max}$ larger than $\Lambda_{\text{crit}}$.

In the following section, we apply the dynamic optimization approach in order to design a time-dependent release program under which local elimination of the wild population can be reached in minimum time. Such program is characterized by the shape of piecewise continuous function $u(t)$ that expresses for each $t$ an instantaneous release rate of sterile males.

\section{Optimal control approach}
\label{sec-oc}

Our primary goal is to define the time-dependent release rate $u^{*}(t)$ that ensures local elimination of the wild population in minimum time $T^{*} \in (0, \infty)$ while also minimizing the total (cumulative) number of sterile males to be released in the target locality during the whole campaign, that is, during the period $[0, T^{*}]$.

Let us suppose that the terminal time of control action $0 < T < \infty$ is set free, and then define several objectives of the decision-making:
\begin{enumerate}
\item
Reaching elimination of the wild population in finite time.
\item
Minimizing the overall time of control action.
\item
Minimizing the underlying costs of control action.
\end{enumerate}

The first objective can be formalized by imposing an endpoint condition
\begin{equation}
\label{endpoint}
 F(T) = \varepsilon,
\end{equation}
where $\varepsilon \to 0^+$ is specified by the decision-maker. Strictly speaking, the endpoint condition \eqref{endpoint} should be of the form $F(T^{*})= 0$. However, given the exponential nature of state equations \eqref{SITsys-a}-\eqref{SITsys-b}, the trajectory of $F(t)$ may only approach zero asymptotically when $t \to \infty$ but it cannot reach this value in finite time. Therefore, reaching elimination of wild females at some finite $0< T< \infty$ can be modeled by \eqref{endpoint} with some small enough $\varepsilon >0$.
Elimination of wild females inevitably entails the elimination of wild male insects, since wild males are progenies of wild females. Therefore, condition (1) ensures the elimination of the total wild population (males and females).

Time minimization (the second objective) can be expressed in the standard way by the identity $T=\int_0^T dt$, while the costs associated with a certain release program $u(t), t \in [0,T]$ (the third objective) are conveyed through the cumulative number of sterile males (CNSM) released during the whole campaign, which is exactly assessed by the following formula:
\begin{equation}
\label{cnsm}
\text{CNSM}(u) = \int \limits_0^{T} u(t) dt.
\end{equation}

Thus, the above-mentioned three objectives can be altogether expressed by the following performance index (or objective functional):
\begin{equation}
\label{of}
\mathcal{J}(u,T)= A_1 \Big(F(T) - \varepsilon \Big)^2 + \int \limits_0^T \left[ A_2 F(t) +  A_3 + \frac{1}{2} A_4 u^2(t)  \right] dt
\end{equation}

The nonnegative coefficients $A_i, i =1,2,3,4$ in \eqref{of} should be adequately chosen in order to reflect the desirable priorities of decision-making. Namely, $A_1$ defines the top priority of reaching elimination in finite time\footnote{Here $A_1$ stands for a penalty for violation of the endpoint condition \eqref{endpoint}, and therefore we suppose that $A_1 \gg A_i, i=2,3,4$.}, $A_2$ characterizes the importance of ``elimination intensity'' of wild females $F(t)$ during the release campaign, $A_3$ expresses the time appreciation $\left( T= \int_0^T dt \right)$, and $A_4$ refers to the costs of control effort (i.e., mass-rearing of sterile males).

It is worthwhile to point out that we assume no linear relationship between the coverage of control interventions and their respective costs. However, we do assume that the \emph{marginal} cost of control action, i. e. $A_4 u(t)$, is proportional to the control effort at each $t \in [0, T ]$ which, in its turn, is expressed by the number sterile males to be released at each $t$. For that reason, the integrand in \eqref{of} is assumed quadratic with respect to control variable $u$. Under such assumption, the ``direct'' costs of a certain release program $u(t), t \in [0,T]$ can be expressed by $A_4 \times \text{CNSM}(u),$ where $\text{CNSM}(u)$ is given by \eqref{cnsm}, while other ``indirect'' costs can be assessed by the length $T$ of the SIT-control intervention. We are aware that this assumption is not very appealing but, at this stage, it is the most convenient one.

Our goal is to find an optimal release program $u^{*}(t) \in [0, u_{\max}], t \in [0, T^{*}]$ and the minimum time $T^{*} \in (0, \infty)$ that jointly minimize the performance index \eqref{of} subject to the wild population dynamics given by \eqref{SITsys}. In other words, we seek to solve the following optimal control problem:
\begin{subequations}
\label{ocp}
\begin{equation}
\label{of-ocp}
\hspace{-12mm} \min \limits_{\scriptsize \begin{array}{c} 0 \leq u(t) \leq u_{\max} \\ 0 < T < \infty \end{array}} \!\!\!
\mathcal{J}(u,T)= A_1 \Big(F(T) - \varepsilon \Big)^2 + \int \limits_0^T \left[ A_2 F(t) +  A_3 + \frac{1}{2} A_4 u^2(t)  \right] dt
\end{equation}
subject to
\begin{align}[left = \empheqlbrace\,]
\label{sys-ocp-a}
\dot M &=r\rho\dfrac{FM}{M+\gamma M_S} \mye{-\beta(M+F)}-\mu_M M, & & \hspace{-1mm} M(0) =M^0 > 0\\
\label{sys-ocp-b}
\dot F &=(1-r)\rho\dfrac{FM}{M+\gamma M_S} \mye{-\beta(M+F)}-\mu_F F, & & F(0) =F^0 > 0\\
\label{sys-ocp-c}
\dot M_S &=u(t)-\mu_S M_S, & & \hspace{-3mm} M_S(0) = M_S^0 \geq 0.
\end{align}
with
\begin{equation}
\label{cont-ocp}
u(\cdot) \in PC, \quad 0 \leq u(t) \leq u_{\max}, \quad \forall \; t \in [0,T].
\end{equation}
\end{subequations}

Existence of solution to the optimal control problem \eqref{ocp} is justified by the following result.
\begin{thm}
\label{th3}
Assume that $\cN_F >1$ and $u_{\max} > \Lambda_{\text{crit}}$. Then optimal control problem \eqref{ocp} has a solution $\big( u^{*}, T^{*} \big)$ such that
\begin{equation*}
\label{of-ocp-max}
\min \limits_{\scriptsize \begin{array}{c} 0 \leq u(t) \leq u_{\max} \\ 0 < T < \infty \end{array}} \!\!\!
\mathcal{J}(u,T)= \mathcal{J} \big( u^{*}, T^{*} \big).
\end{equation*}
\end{thm}

\begin{proof}
To prove this theorem, we employ the classical approach based on the Filippov-Cesari existence theorem (thoroughly described in \cite{Fleming1975,Seierstad1987}). Roughly speaking, we have to show that our optimal control problems \eqref{ocp} fulfils the following \emph{sufficient} conditions for existence of an optimal solution:

\begin{enumerate}
\item[(i)]
Solution $\mathbf{X}(t):= \big( M(t), F(t), M_S(t) \big)^{'}$ of the dynamical system \eqref{sys-ocp-a}-\eqref{sys-ocp-c} is well-defined and unique for each admissible $u(t) \in [0,u_{max}], t \geq 0 $, while the set of all $\mathbf{X}(t)$ is non-empty and bounded  $\forall \; u(t) \in [0,u_{max}]$ and $ \forall \; t \geq 0$.
\item[(ii)]
The sets of all initial and terminal states $\mathbf{X}(0)=\big( M(0),F(0),M_S(0) \big)$, $\mathbf{X}(T)=\big( M(T^{*}),F(T^{*}),M_S(T^{*}) \big)$ are closed and bounded in $\mathbb{R}^3_+$.
\item[(iii)]
The control set $[0,u_{\max}]$ is closed, bounded and convex in $\mathbb{R}$.
\item[(iv)]
The right-hand sides of the dynamical system \eqref{sys-ocp-a}-\eqref{sys-ocp-c} are linear in the control variable $u$.
\item[(v)]
The terminal-state function $ \phi(\mathbf{X}) :=A_1 (F - \varepsilon)^2$ is continuous in its arguments.
\item[(vi)]
The integrand of \eqref{of-ocp} is convex in $u$ and satisfies the coercivity condition $A_2 F(t) + A_3 + \dfrac{1}{2} A_4 u_{1}^{2}(t)  \geq \mathcal{K}_1 \big\| u \big\|^{\alpha} - \mathcal{K}_2$ with some constants $\mathcal{K}_1 >0, \ \alpha > 1,$ and $\mathcal{K}_2$.
\end{enumerate}	

Items (i)-(ii) have been already corroborated in Section~\ref{sec-model} where it was established that dynamical system \eqref{sys-ocp-a}-\eqref{sys-ocp-c} has a unique solution for any $u(t) \in [0,u_{\max}]$ that exists for all $t \geq 0$. Moreover, the trajectories $M(t), F(t), M_S(t)$ remain bounded for all $t \geq 0$ since they belong to the compact absorbing set
\[ \mathcal{Z} := \left\{ \big( M(t), F(t) \big) \in \mathcal{D}, 0 \leq M_S(t) \leq \dfrac{u_{\max}}{\mu_S}, \ \forall \: t \geq 0  \right\} \]
if engendered by $\big( M(0), F(0), M_S(0) \big) \in \mathcal{Z}$. In the above definition of $\mathcal{Z}$, $\mathcal{D}$ refers to \eqref{d-set}.

The credibility of items (iii)-(v) is beyond doubt, while item (vi) is fulfilled with $\mathcal{K}_1=\dfrac{1}{2} A_4, \alpha=2$ and $\mathcal{K}_2=0$ for all $t \geq 0$. This completes the proof.
\end{proof}

Formal solution of the optimal control problem \eqref{ocp} can be achieved by applying the Pontryagin maximum principle that constitutes the necessary condition for optimality of $\big( u^{*}, T^{*} \big),$ which exists in virtue of Theorem~\ref{th3}. In particular, we are interested in the variant of maximum principle applicable to optimal control problems with free terminal time, concisely described by Lenhart
and Workman \cite{Lenhart2007}. The Hamiltonian function $H(M,F,M_S,u,\lambda_1,\lambda_2,\lambda_3): \: \mathcal{Z} \times \mathbb{R}^3 \mapsto \mathbb{R}$ associated with the optimal control problem \eqref{ocp} is defined as
\begin{align}
\label{ham}
H(\mathbf{X}, u, \boldsymbol{\lambda}) &=  - A_2 F -  A_3 - \frac{1}{2} A_4 u^2 + \lambda_1 \left[ r\rho\dfrac{FM}{M+\gamma M_S} \mye{-\beta(M+F)}-\mu_M M \right] \notag \\
&+ \lambda_2 \left[ (1-r) \rho\dfrac{FM}{M+\gamma M_S} \mye{-\beta(M+F)}-\mu_M M \right] + \lambda_3 \Big[ u(t)-\mu_S M_S \Big]
\end{align}
where $\boldsymbol{\lambda}:= \big( \lambda_{1}, \lambda_{2}, \lambda_3 \big)^{'}$ can be viewed as a vector of Lagrange multipliers linked to differential constraints \eqref{sys-ocp-a}-\eqref{sys-ocp-c}. On the other hand, $\lambda_{i}=\lambda_{i}(t), i=1,2,3$ are time-varying real \emph{adjoint} functions that express the marginal variations in the value of objective functional $\mathcal{J}(u,T)$ induced by changes in the current values of state variables $\big( M(t), F(t), M_S(t) \big)$ (in the ``component-by-component'' sense). Thus, the current values of $\lambda_{i}=\lambda_{i}(t), i=1,2,3$ reflect additional benefits or costs associated with changes in $\big( M(t), F(t), M_S(t) \big)$ and they are necessary elements of the Pontryagin maximum principle \cite{Lenhart2007}, which is formulated as follows.

Let $\big( u^{*}, T^{*} \big)$ be an optimal pair for \eqref{ocp} in the sense that $u^{*}(t)$ is a piecewise continuous real function with domain $[0, T^{*}]$ and range $[0, u_{\max}]$ and
\[ \mathcal{J}\big( u^{*}, T^{*} \big) \leq  \mathcal{J}(u,T) \]
for all other controls $u$ and times $T$. Let $\mathbf{X}^{*}(t)= \big( M^{*}(t), F^{*}(t), M_S^{*}(t) \big)^{'}$ be the corresponding state defined for all $t \in [0, T^{*}]$. Then there exists a piecewise differentiable adjoint function $\boldsymbol{\lambda}: [0, T^{*}] \mapsto \mathbb{R}^3$  satisfying the \emph{adjoint} differential equation
\begin{equation}
\label{adj-sys}
\dot{\boldsymbol{\lambda}} = - H_{\mathbf{X}} \big(\mathbf{X}^{*}, u^{*}, \boldsymbol{\lambda} \big)
\end{equation}
with transversality conditions
\begin{equation}
\label{t-con}
 \lambda_{1} (T^{*})=0, \quad \lambda_{2} (T^{*})= - \phi_{\mathbf{X}}\Big|_{\mathbf{X}=\mathbf{X}(T^{*})} = -2A_1 \big( F^{*}(T^{*}) - \varepsilon \big), \quad \lambda_{3} (T^{*})=0
\end{equation}
while $0 < T^{*} < \infty$ fulfills the time-optimality condition
\begin{equation}
\label{min-t}
 H \big(\mathbf{X}^{*}(T^{*}), u^{*}(T^{*}), \boldsymbol{\lambda}(T^{*}) \big) =0.
\end{equation}
Furthermore, the Hamiltonian \eqref{ham} has a critical point (maximum\footnote{It is straightforward to verify that $H_{uu} = -A_4 < 0$ for all $u$. }) at $u = u^{*}$, and therefore
\begin{equation}
\label{pmax}
H \big(\mathbf{X}^{*}(t), u^{*}(t), \boldsymbol{\lambda}(t) \big) \geq  H \big(\mathbf{X}^{*}(t), u(t), \boldsymbol{\lambda}(t) \big)
\end{equation}
for any $u(t): [0,T^{*}] \mapsto [0, u_{\max}]$ and almost for all $t \in [0,T^{*}]$.

According to \cite{Lenhart2007}, condition \eqref{pmax} can be written more precisely as
\begin{equation}
\label{optcon}
\left.  \begin{array}{ccc}
u^{*}(t)=0 & \text{if} & H_u < 0\\
0 < u^{*}(t)< u_{\max} & \text{if} & H_u=0 \\
u^{*}(t)=u_{\max} & \text{if} & H_u > 0
\end{array} \right\}
\end{equation}
that leads to a more compact form of $u^{*}(t)$ also known as the \emph{characterization} of the optimal control:
\begin{equation}
\label{char-u}
u^{*}(t) = \max \left\{ 0, \min \left\{ \frac{1}{A_4} \cdot \lambda_3 (t), u_{\max} \right\} \right\}.
\end{equation}

It is worth pointing out that the optimality conditions \eqref{optcon} are rather meaningful and provide interesting insights regarding the benefits and costs of control interventions. Namely, the condition
\[ H_u = - A_4 u + \lambda_3 = 0 \]
implies that, under optimal release program $u^{*}$ and at each $t$, the marginal cost of control intervention (given by the term $A_4 u$) must be equal to its marginal benefit (expressed by the term $\lambda_3$). If the marginal cost of $u^{*}$ is higher than its marginal
benefit (that is, $H_u < 0$ in \eqref{optcon}) then it is optimal not to implement this program at all, and therefore we set $u^{*}(t)=0$.
On the other hand, if the marginal cost of $u^{*}$  is lower than its marginal benefit (that is, $H_u > 0$ in \eqref{optcon}), then it is optimal to use all available resources, and therefore we set $u^{*}(t)=u_{\max}.$

Using the characterization of $u^{*}$ given by \eqref{char-u}, the original optimal control problem \eqref{ocp} can be reduced to a
two-point boundary value problem. This boundary value problem (also known as  \emph{optimality system}) is composed by six differential equations with six endpoint conditions, namely:
\begin{itemize}
\item
three original equations with underlying initial conditions  \eqref{sys-ocp-a}-\eqref{sys-ocp-c} where $u(t)$ is replaced by its characterization \eqref{char-u};
\item
three adjoint equations \eqref{adj-sys} with transversality conditions \eqref{t-con} where $u(t)$ is also replaced by its characterization \eqref{char-u}.
\end{itemize}
The optimal time $0 < T^{*} < \infty$ is defined by the time-optimality condition \eqref{min-t}.

Due to non-linearity, complexity, and high dimension of the optimality system described above, it can only be solved numerically, and the forthcoming section is fully devoted to this issue, with a particular focus on \textit{Aedes spp} parameters, given in Table \ref{tab1}.

\section{Numerical simulations and discussion}
\label{sec-num}

\subsection{Preliminaries}

Our numerical simulations are focused on the "worst scenario'', that is, supposing that wild mosquitoes have equilibrium densities at the commencement of control intervention ($t=0$), while sterile males are not present in the target locality ($M_S(0)=0$). Bliman \emph{et al} \cite{Bliman2019} had deduced that the positive steady state $E^{\sharp}=\big(M^{\sharp},F^{\sharp} \big)$ of the sex-structured system \eqref{SITsys} (or \eqref{SIT-reduced}) with $M_S(t) \equiv 0$ has coordinates
\begin{equation}
\label{steady}
M^{\sharp} := \frac{\cN_M}{\cN_F + \cN_M} \frac{1}{\beta} \ln \cN_F \quad \text{and} \quad F^{\sharp} := \frac{\cN_F}{\cN_F + \cN_M} \frac{1}{\beta} \ln \cN_F,
\end{equation}
where the basic offspring numbers $\cN_F, \: \cN_M$ satisfy the relationships \eqref{bon}. Thus, initial conditions for the dynamical system \eqref{sys-ocp-a}-\eqref{sys-ocp-c} are defined as
\[ M(0)=M^{\sharp}, \quad F(0)=F^{\sharp}, \quad M_S(0)=0. \]

For the parameter values given in the last column of Table~\ref{tab1}, we have that $\cN_F \approx 75.83, \: \cN_M \approx 56.87$, while the mosquito densities at equilibrium $E^{\sharp}=\big(M^{\sharp},F^{\sharp} \big)$ are $M^{\sharp} \approx 5.19 \times 10^3$ and $M^{\sharp} \approx 6.93 \times 10^3$ individuals per hectare (indv/ha).

We assume that elimination of wild population is reached when the endpoint condition \eqref{endpoint} is fulfilled with $\varepsilon =10^{-1}$ meaning that no female mosquitoes are left in the target locality ($F^{*}(T^{*}) < 1$).

The value of $u_{\max}$ must be set greater than $\Lambda_{\text{crit}} \approx 1.29 \times 10^3$ (see Theorems~\ref{th1} and ~\ref{th2} in Section~\ref{sec-model}). On the one hand, $u_{\max}$ should not be too high in order to avoid extra-large releases and to maintain feasible  the overall cost of the control intervention. On the other hand, $u_{\max}$ must be large enough in order to potentiate a faster elimination of wild populations. Therefore, we set
\[ u_{\max}=2.5 \times 10^3 \]
as the maximum release capacity in all ensuing simulations. This choice is also linked to the production capacity of the sterile insect factory, and the proper value of $u_{\max}$ must be defined by practitioners.

The coefficients $A_i, i =1,2,3,4$ in \eqref{of-ocp} must account for scaling of all quantities involved in the objective functional. In this study, we adopt the following scaling for $A_i, i=1,2,3,4$:
\begin{equation}
\label{coef}
A_1= \frac{\mathcal{P}_1}{F^{\sharp}}, \quad A_2 := \frac{\mathcal{P}_2}{F^{\sharp}}, \quad A_3= \frac{\mathcal{P}_3}{365}, \quad A_4=\frac{\mathcal{P}_4}{u_{\max}},
\end{equation}
where $F^{\sharp}$ refers to the steady-state value of wild females \eqref{steady}, time is scaled to years ($=365$ days), and the control effort is normalized by the maximal release capacity $u_{\max}$. In the gauged values \eqref{coef}, the coefficients $\mathcal{P}_i, i=1,2,3,4$ define ``pure'' priorities of the decision-making, which are dimensionless.

Our top priority is to reach elimination of wild mosquitoes in minimum time; therefore, the largest value is assigned to $\mathcal{P}_1$ thus seeking to fulfill the endpoint condition \eqref{endpoint} and to obtain that $A_1 \phi \big( \mathbf{X}^{*}(T^{*}) \big) = 0$. On the other hand,  the lowest priority  can be assigned to $\mathcal{P}_4$ that stands for the rearing (unit) cost of one sterile males.

The intensity of wild mosquitoes elimination (expressed by $\mathcal{P}_2$ in the objective functional \eqref{of-ocp}) is also essential, and a higher value of $\mathcal{P}_2$ must help to continue elimination even when the current value of $F(t)$ becomes relatively small. The role of time minimization ($\mathcal{P}_3$) should be further explored by considering three options: (1) $\mathcal{P}_3=0$, (2) $\mathcal{P}_4 < \mathcal{P}_3 < \mathcal{P}_2, $ and (3) $\mathcal{P}_3 > \mathcal{P}_2$.

Summarizing the above-mentioned considerations, we consider the following values of ``pure'' priority coefficients:
\begin{equation}
\label{pure-p}
\mathcal{P}_1 = 10^{10}, \quad  \mathcal{P}_2= 10^4, \quad \mathcal{P}_3 \in \big\{ 0, 10^3, 10^5 \big\}, \quad \mathcal{P}_4=1.
\end{equation}

\subsection{Continuous-time optimal release programs}
\label{subsec-con}

To solve numerically the optimality systems described in the closing part of Section~\ref{sec-oc} together with time optimality condition \eqref{min-t} for free terminal time, we employ the GPOPS-II solver\footnote{For more information regarding the GPOPS-II solver please visit http://gpops2.com/. A concise description of this solver and its capacities is also available in \cite[Appendix B]{Campo2018}.} designed for the MATLAB platform. This solver implements an adaptive numerical technique based on the \emph{direct orthogonal collocation}, which is also known as \emph{Radau pseudospectral method} \cite{Garg2011,Patterson2014}.

Additionally, GPOPS-II solver automatically scales all input intervals $[0,T]$ to the range $[-1,1]$, thus it suffices to hold its standard numerical tolerance of $10^{-5}$ with regards to internal scaling.

\begin{figure}[t!]
\begin{center}
\begin{tabular}{ccc}
& Optimal control profiles & \\
& & \\
\includegraphics[width=.33\textwidth]{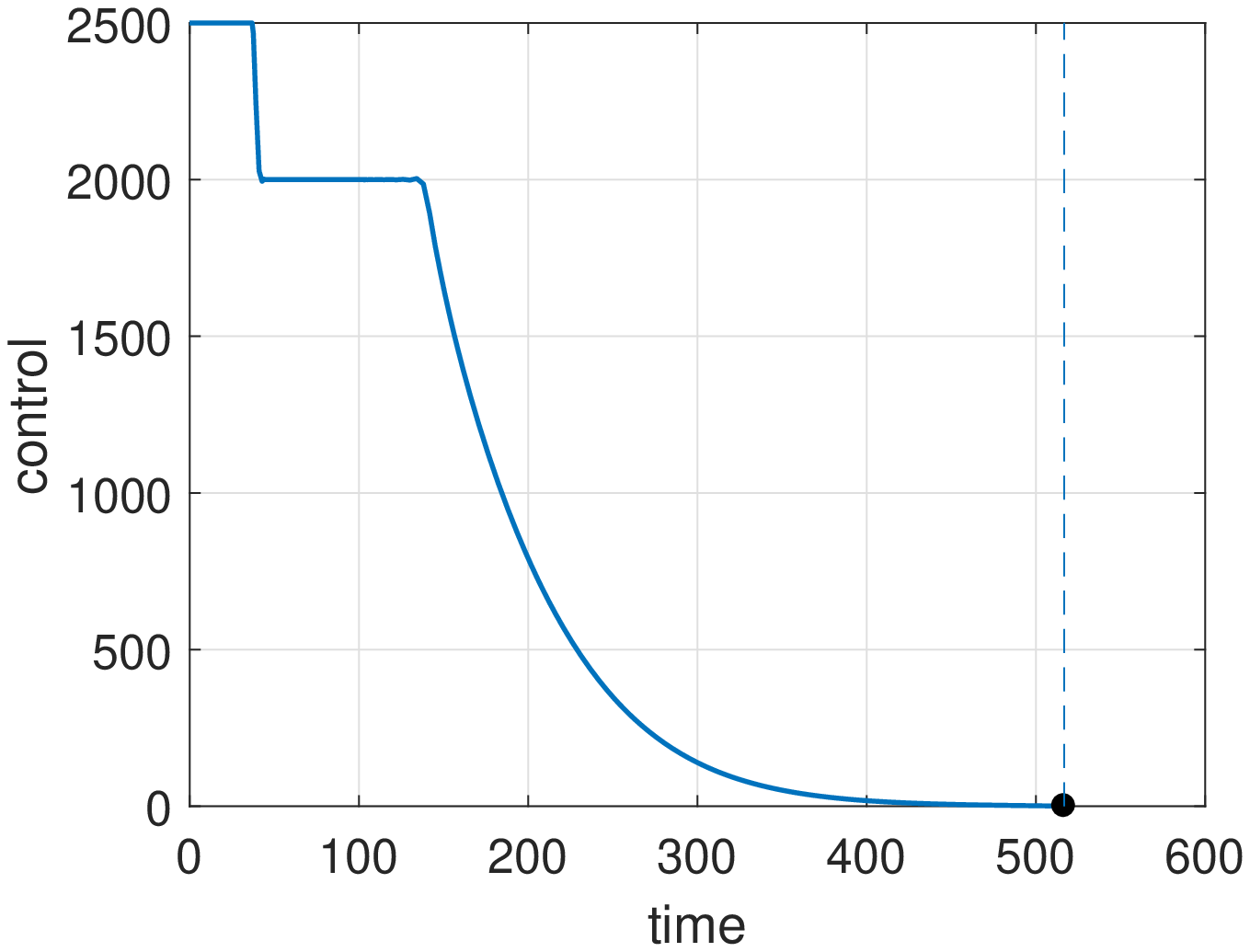} &
\includegraphics[width=.33\textwidth]{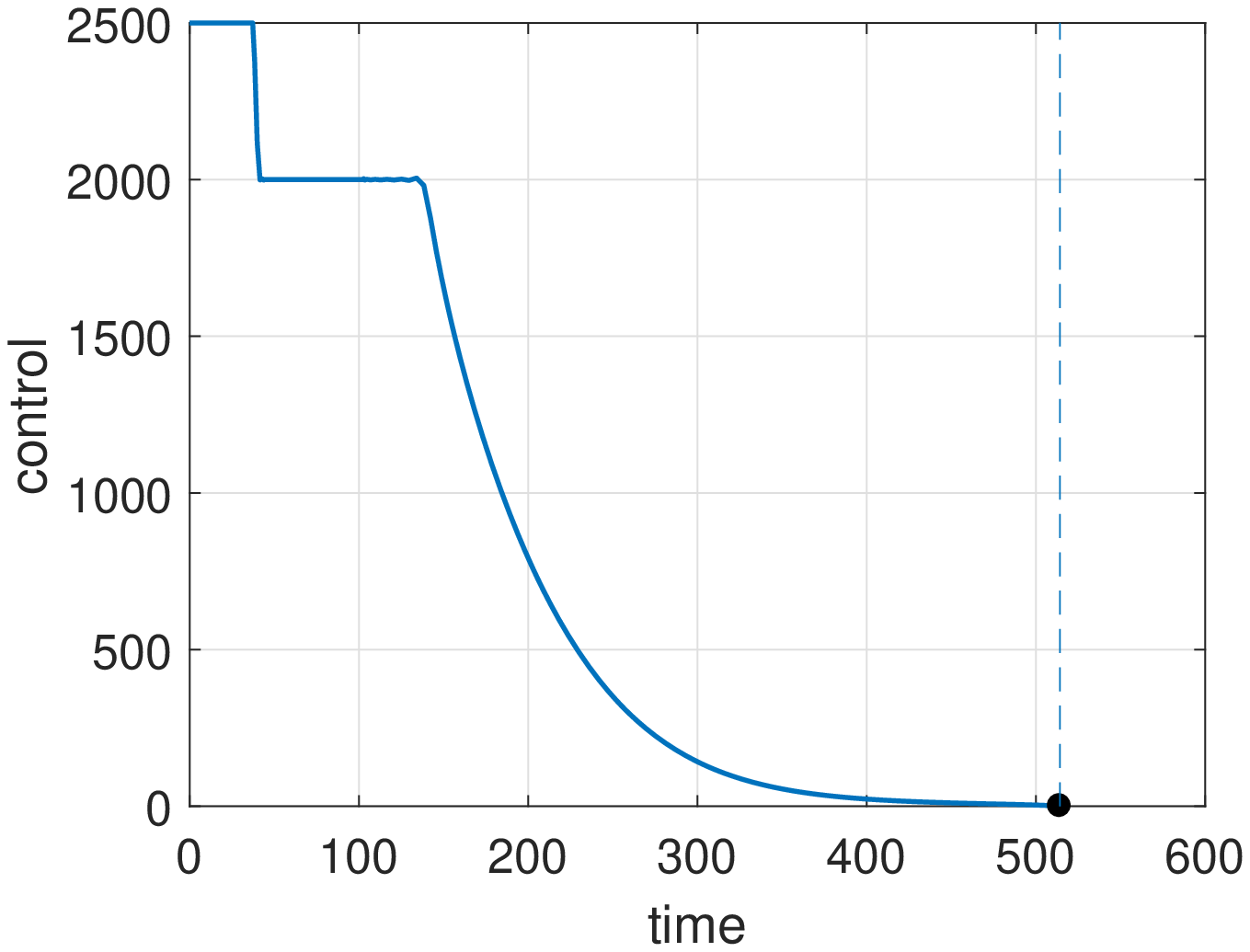} &  \includegraphics[width=.33\textwidth]{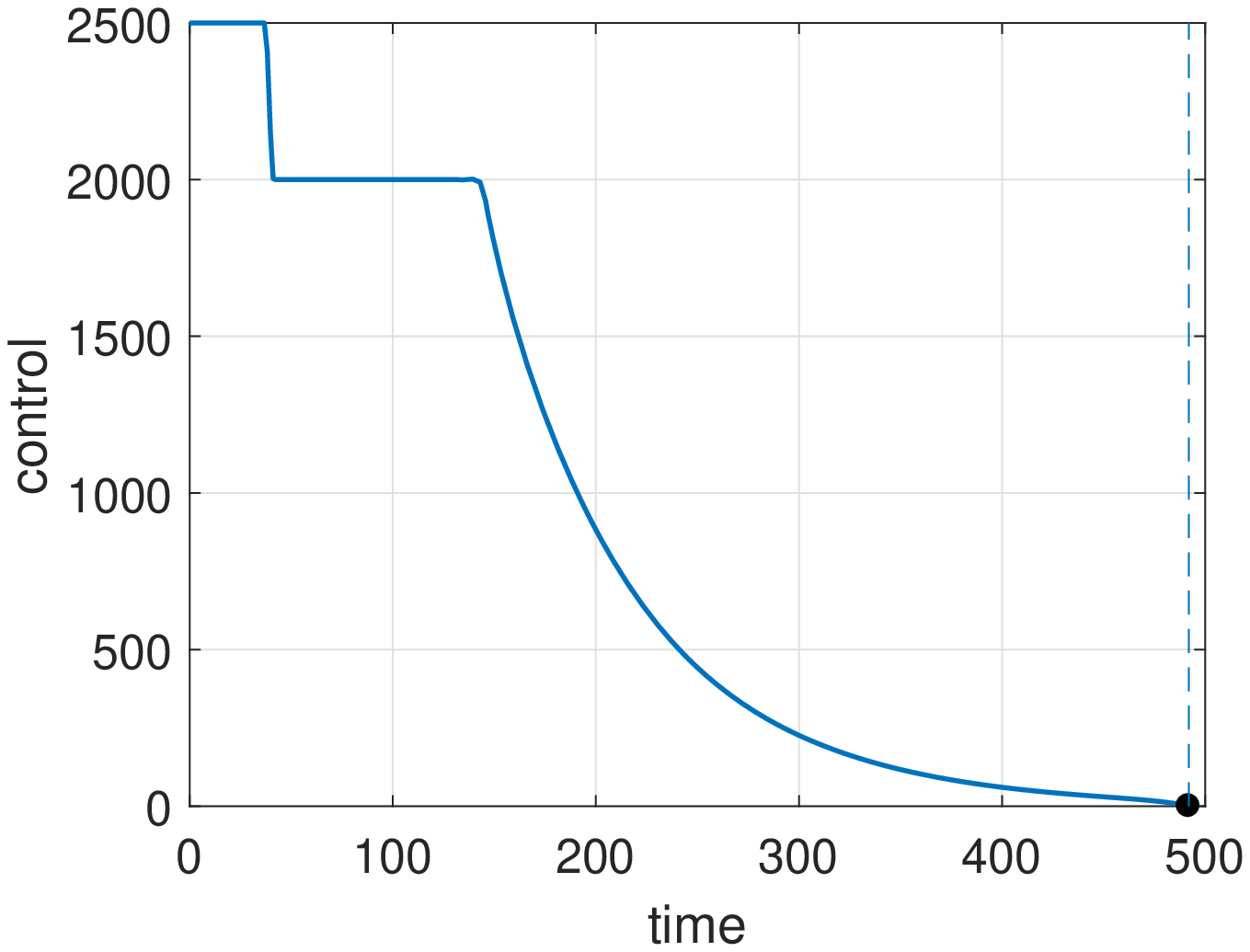} \\
& & \\
& Trajectories of optimal states & \\
& & \\
\includegraphics[width=.33\textwidth]{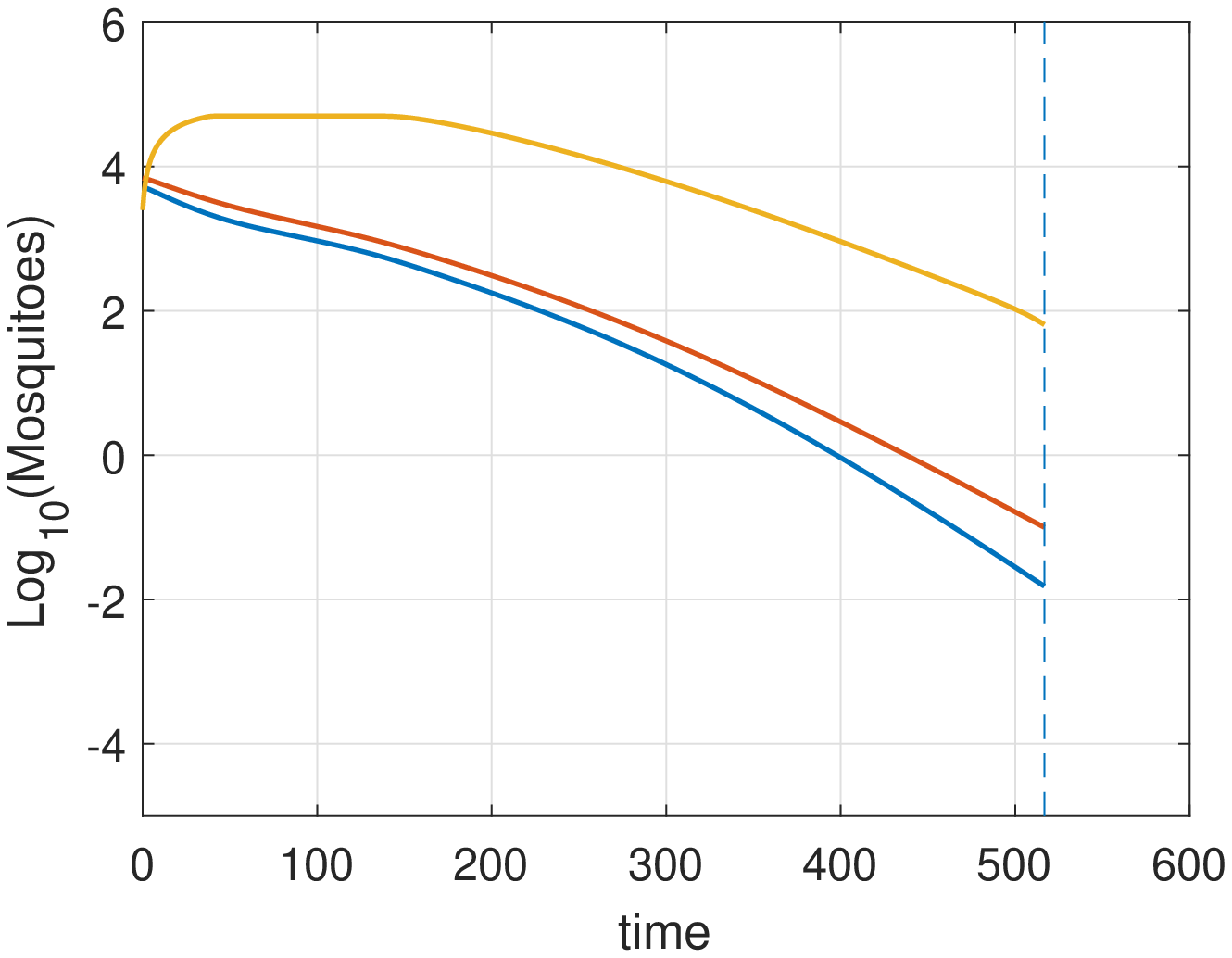} &  \includegraphics[width=.33\textwidth]{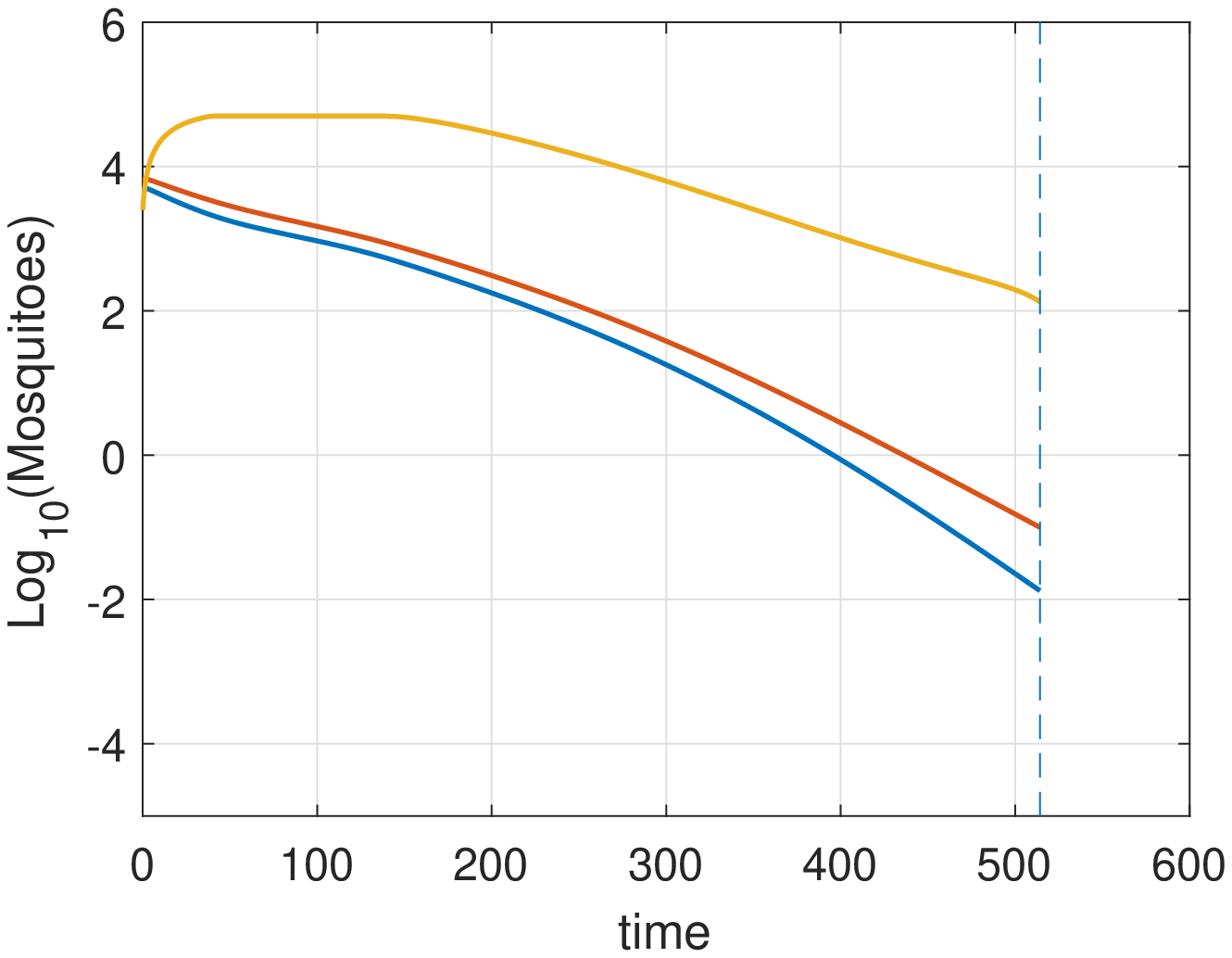}  & \includegraphics[width=.33\textwidth]{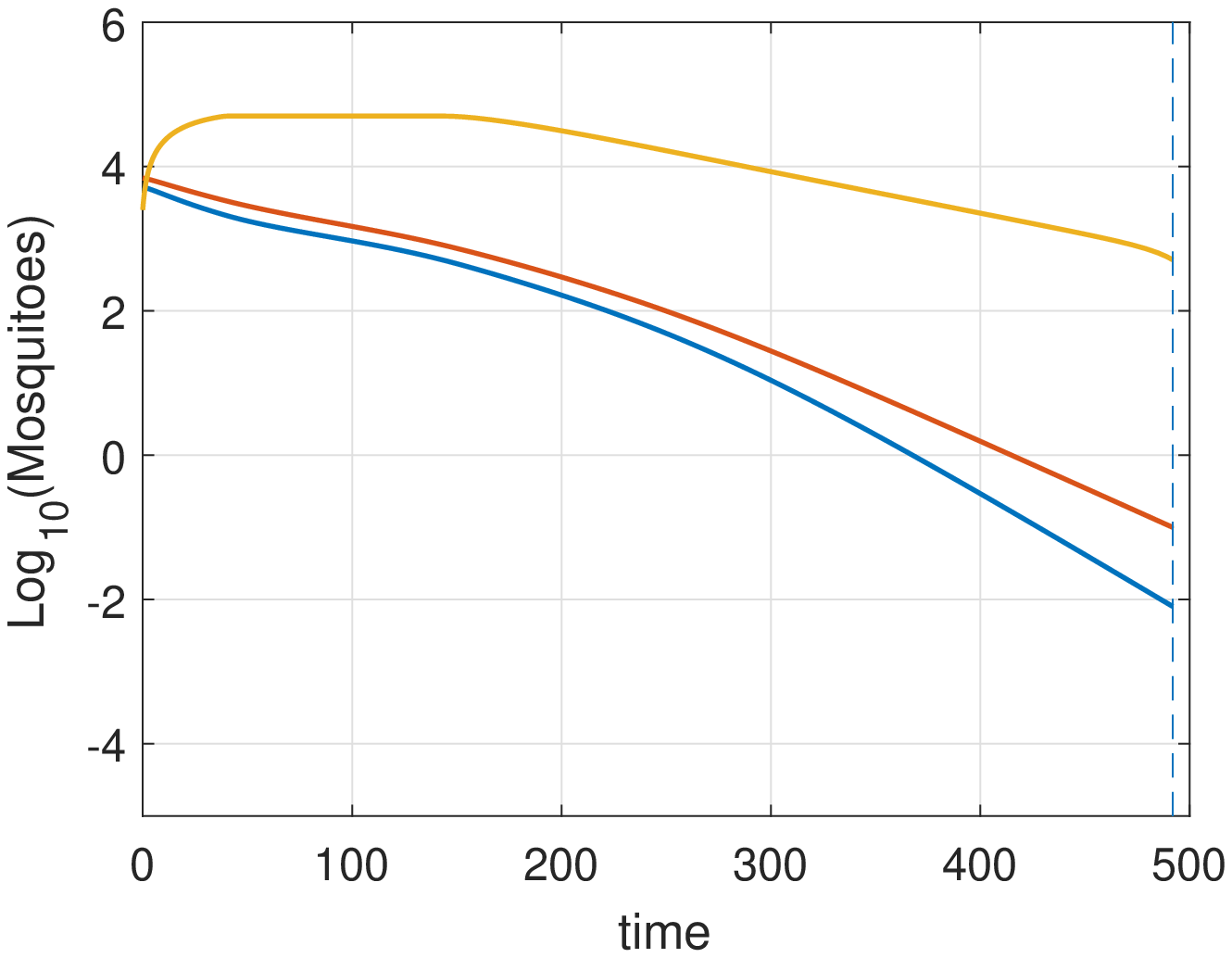} \\
 $\mathcal{P}_3=0, \; T^{*}=517$ days;	 &  $\mathcal{P}_3=10^3, \;  T^{*}=515$ days & $\mathcal{P}_3=10^5, \; T^{*}=492$ days \\
$\text{CNSM}(u^{*})= 421,640$ indv/ha & $\text{CNSM}(u^{*})= 422,730$ indv/ha & $\text{CNSM}(u^{*})= 446,200 $  indv/ha
\end{tabular}
\end{center}
\caption{\emph{Upper row:} Optimal release programs $u^{*}(t), \: t \in [0,T^{*}]$ for $\mathcal{P}_3 \in \big\{ 0, 10^3, 10^5 \big\}$. \emph{Lower row:} Time evolution of mosquito populations under respective $u^{*}(t)$: wild males $\log_{10} M(t)$ (blue-colored curves), wild females $\log_{10} F(t)$ (red-colored curves), and sterile males $\log_{10} M_S(t)$ (yellow-colored curves).  }
\label{fig1}
\end{figure}

Figure~\ref{fig1} presents the results of numerical solutions of the optimal control problem \eqref{ocp} for three sets of ``pure'' priority values assigned in \eqref{pure-p}. The upper row of Figure~\ref{fig1} displays very similar structure of the optimal release programs $u^{*}(t)$ for all three cases, what naturally shows off their robustness. However, as $\mathcal{P}_3$ increases, the overall number of sterile insects needed for the program's implementation grows larger, while the overall time $T^{*}$ of the control intervention slowly shortens down. This provides interesting insights for decision-makers by  illustrating the existence of a certain tradeoff between the overall duration of the elimination campaign and its underlying costs.

Additionally, it is worthwhile to note that optimal release programs are suspended exactly at $t=T^{*}$, that is, when the endpoint condition \eqref{endpoint} is finally met. The latter is induced by the transversality conditions \eqref{t-con}, according to which no additional benefit is expected after $T^{*}$.

Even though all three optimal release programs $u^{*}(t), \: t \in [0,T^{*}]$ displayed in Figure~\ref{fig1} have a very clear and straightforward structure, their tangible implementation may become unfeasible since the continuous-time releases are hardly practicable.

In the following subsection, we propose a more realistic type of release programs which are inspired by the optimal solutions $u^{*}(t)$ and, due to that reason, they are further referred to as ``suboptimal'' strategies based on impulsive releases of sterile males.

\subsection{Practical applications: suboptimal impulsive release programs}
\label{subsec-imp}

From the practical standpoint, it is more realistic to consider periodic releases of sterile male mosquitoes every $n \tau$ days, where $n \in \mathbb{N}$ until meeting the endpoint condition \eqref{endpoint}. In fact, Bliman \emph{et al} \cite{Bliman2019} studied an extension of SIT-control model \eqref{SITsys}, namely
\begin{subequations}
\label{SITsys-imp}
\begin{align}[left = \empheqlbrace\,]
\label{SITsys-imp-a}
\dot M &=r\rho\dfrac{FM}{M+\gamma M_S} \mye{-\beta(M+F)}-\mu_M M, & & \hspace{-1mm} M(0) =M^0 > 0, \\
\label{SITsys-imp-b}
\dot F &=(1-r)\rho\dfrac{FM}{M+\gamma M_S} \mye{-\beta(M+F)}-\mu_F F, & & F(0) =F^0 > 0, \\
\label{SITsys-imp-c}
\dot M_S &= -\mu_S M_S, & & \forall \: t \in \bigcup_{n \in \mathbb{N}} \big( n \tau, (n+1)\tau \big), \\
\label{SITsys-imp-d}
M_S \big( n\tau^{+} \big) &= \tau \Lambda_n + M_S \big( n\tau^{-} \big), & & n \in \mathbb{N},
\end{align}
\end{subequations}
where the piecewise continuous function $M_S(t)$ has jumps at $t=n \tau$, $M_S \big( n\tau^{+} \big)$ denote its right and left limits at $t=n \tau$, and $\Lambda_n$ stands for a ``per day'' number of sterile male mosquitoes to be released in the target locality. The relationship \eqref{SITsys-imp-d} models that at $t=n \tau$, a lump quantity  $\tau \Lambda_n$ of sterile insects is emitted to the wild population at once, whereas on a union of open intervals $\big(n \tau, (n + 1) \tau \big), n \in \mathbb{N}$ the mosquito populations $M(t), F(t),$ and $M_S(t)$ evolve according to Eqs. \eqref{SITsys-imp-a}-\eqref{SITsys-imp-c}.

It is worth pointing out that Bliman \emph{et al} \cite{Bliman2019} have proposed and studied different strategies for implementation of impulsive releases of sterile males that are engendered by the choice of $\Lambda_n$ in \eqref{SITsys-imp-d}. Namely, the \emph{open-loop} (or \emph{feedforward}) SIT-control strategies for impulsive releases have been designed by assuming $\Lambda_n=\Lambda$ constant for all $n \in \mathbb{N}$. Alternatively, choosing each $\Lambda_n$ in accordance with periodic measurements (either synchronized with releases or more sparse) of the wild population size has enabled the design of the \emph{closed-loop} (or \emph{feedback}) periodic impulsive strategies for SIT-control. Moreover, the combination of two above-mentioned approaches has resulted in the design of mixed (open/closed-loop) strategies for periodic impulsive SIT-control, and this control mode have actually rendered the best outcome expressed through a shorter overall time needed to reach elimination, smaller cumulative number of sterile insects to be released, and fewer number of releases to be effectively carried out during a whole SIT-control campaign.

On the other hand, mixed (open/closed-loop) strategies for periodic impulsive SIT-control still require for real-time assessments of wild population sizes $M(t)$ and $F(t)$, and this protocol is time-consuming and costly, even with sparse measurements carried out every $p\tau$ days where $p=2, 3, \ldots$. When $p=1$, it is understood that the measurements are synchronized with the releases (the reader may find further details regarding the synchronized and sparse measurements in \cite{Bliman2019}).

In this paper, we propose another type of periodic impulsive SIT-control strategies that are expressly feedforward (or open-loop) and thus require no real-time assessments of the wild population sizes. Their design is performed on the grounds of continuous-time numerical solutions $u^{*}(t), t \in [0, T^{*}]$ to the optimal control problem \eqref{ocp} (obtained in Subsection~\ref{subsec-con}), and for this reason we call them ``suboptimal''. Contrary to the continuous case (Section ~\ref{sec-oc}) and the one considered by Bliman \emph{et al} \cite{Bliman2019}, we cannot rigorously prove the existence of an optimal solution,  neither to show the convergence of $F(t)$ to $\epsilon$ in a finite time. Nevertheless, we explore this case numerically.

First, we define the extension of the optimal control program $u^{*}(t)$ as
\begin{equation}
\label{uhat}
\hat{u}^{*}(t) = \left\{ \begin{array}{rcl}
u^{*}(t), & \text{if} & t \in [0, T^{*}] \\ 0, & \text{if} & t > T^{*} \end{array} \right.
\end{equation}
For a chosen period $\tau$, we then construct a sequence of impulses $\{ U_n \}$ using the following rule
\begin{equation}
\label{Un}
U_n = \max_{t \in \big[ n \tau, (n+1)\tau \big] } \hat{u}^{*}(t), \quad n \in \mathbb{N}.
\end{equation}
This sequence $\{ U_n \}$ converges to zero and reaches this value in some finite time $\hat{T}^{*} < \infty$. In view of the relationship \eqref{Un}, it is expected that the overall number of released sterile insects during the whole SIT-control campaign be a bit higher under suboptimal impulsive releases than under continuous-time optimal release programs $u^{*}(t)$. On the other hand, the endpoint condition \eqref{endpoint} may be reached a bit sooner under suboptimal impulsive releases than under continuous-time optimal release programs $u^{*}(t)$, and therefore $\hat{T}^{*} \leq T^{*}$. This is exactly what numerical experiments illustrate --- see results displayed in Figures~\ref{fig2} and~\ref{fig3} for $\tau=7$ days and $\tau=14$ days, respectively.

\begin{figure}[t!]
\begin{center}
\begin{tabular}{ccc}
& suboptimal impulsive & \\
& release strategies & \\
\includegraphics[width=.33\textwidth]{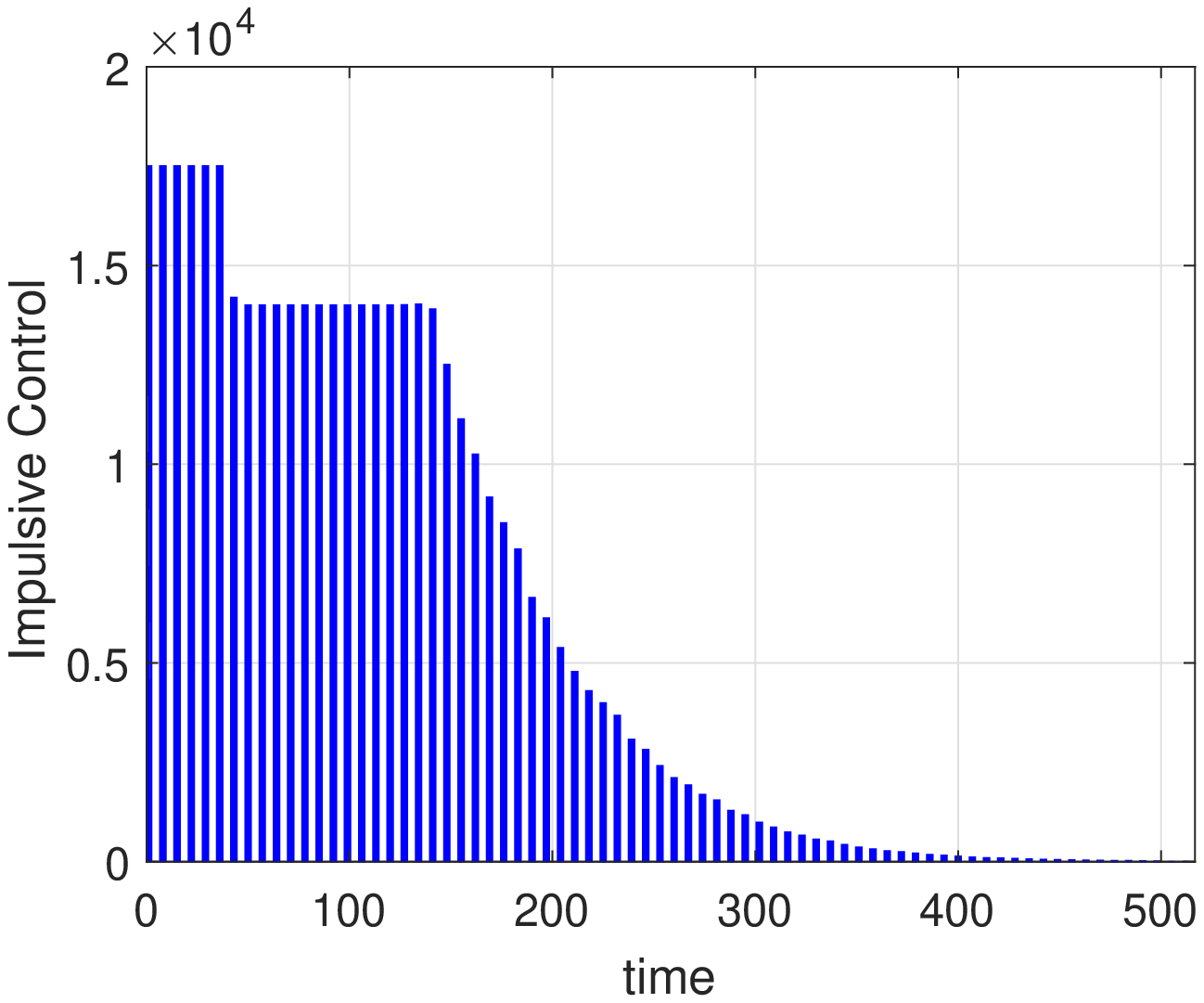} &
\includegraphics[width=.33\textwidth]{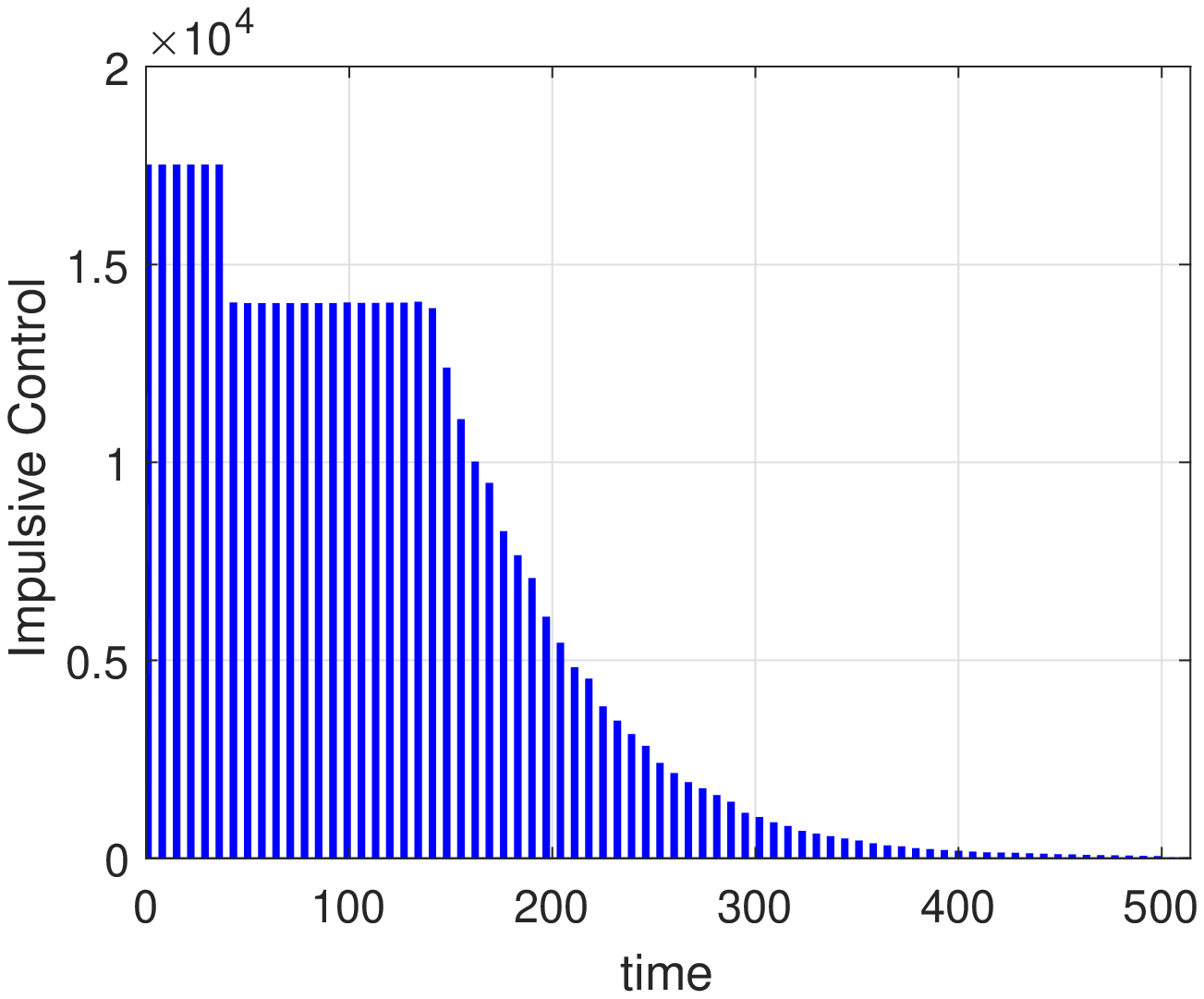} &  \includegraphics[width=.33\textwidth]{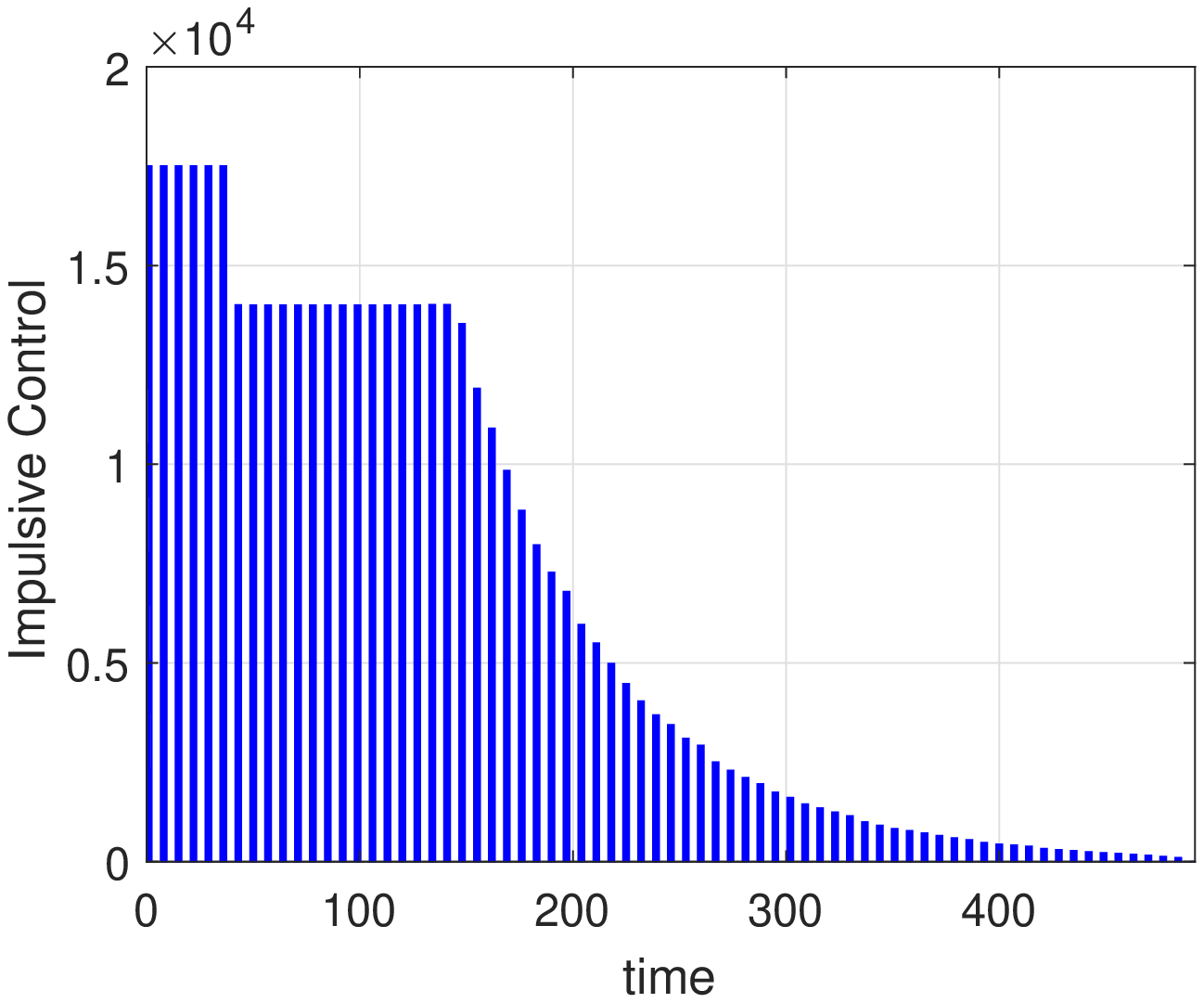} \\
& & \\
& Trajectories of suboptimal states & \\
& & \\
\includegraphics[width=.33\textwidth]{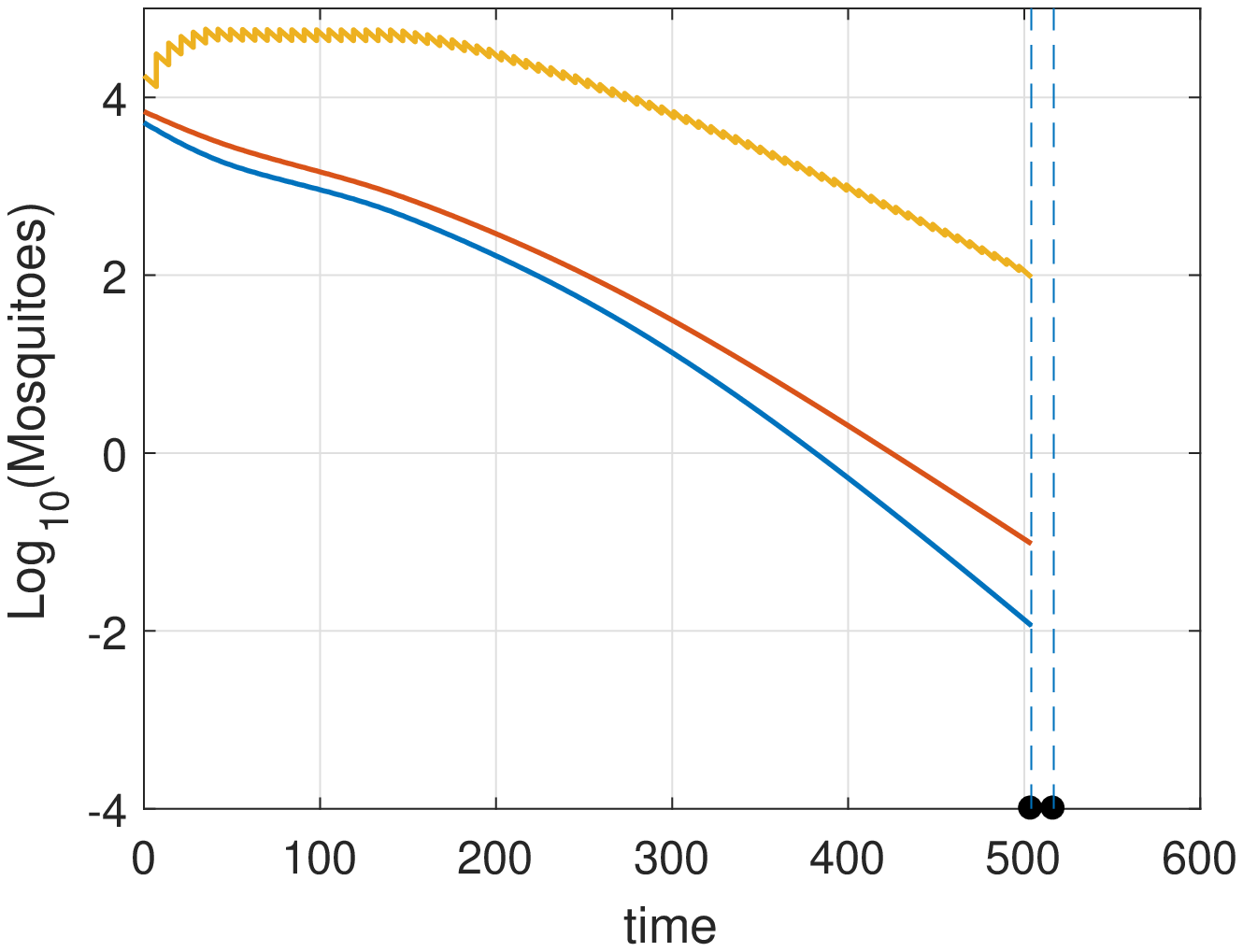} &  \includegraphics[width=.33\textwidth]{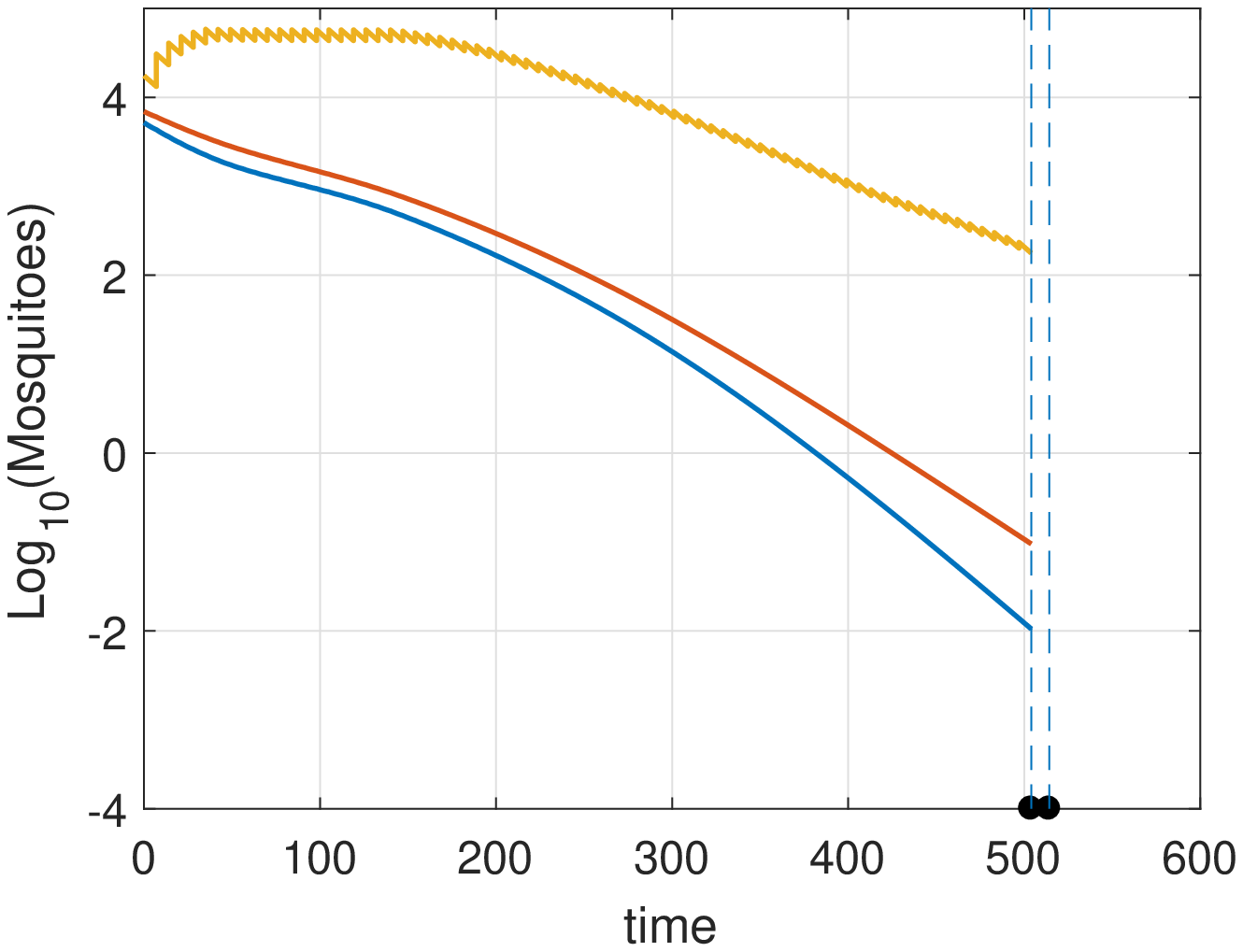}  & \includegraphics[width=.33\textwidth]{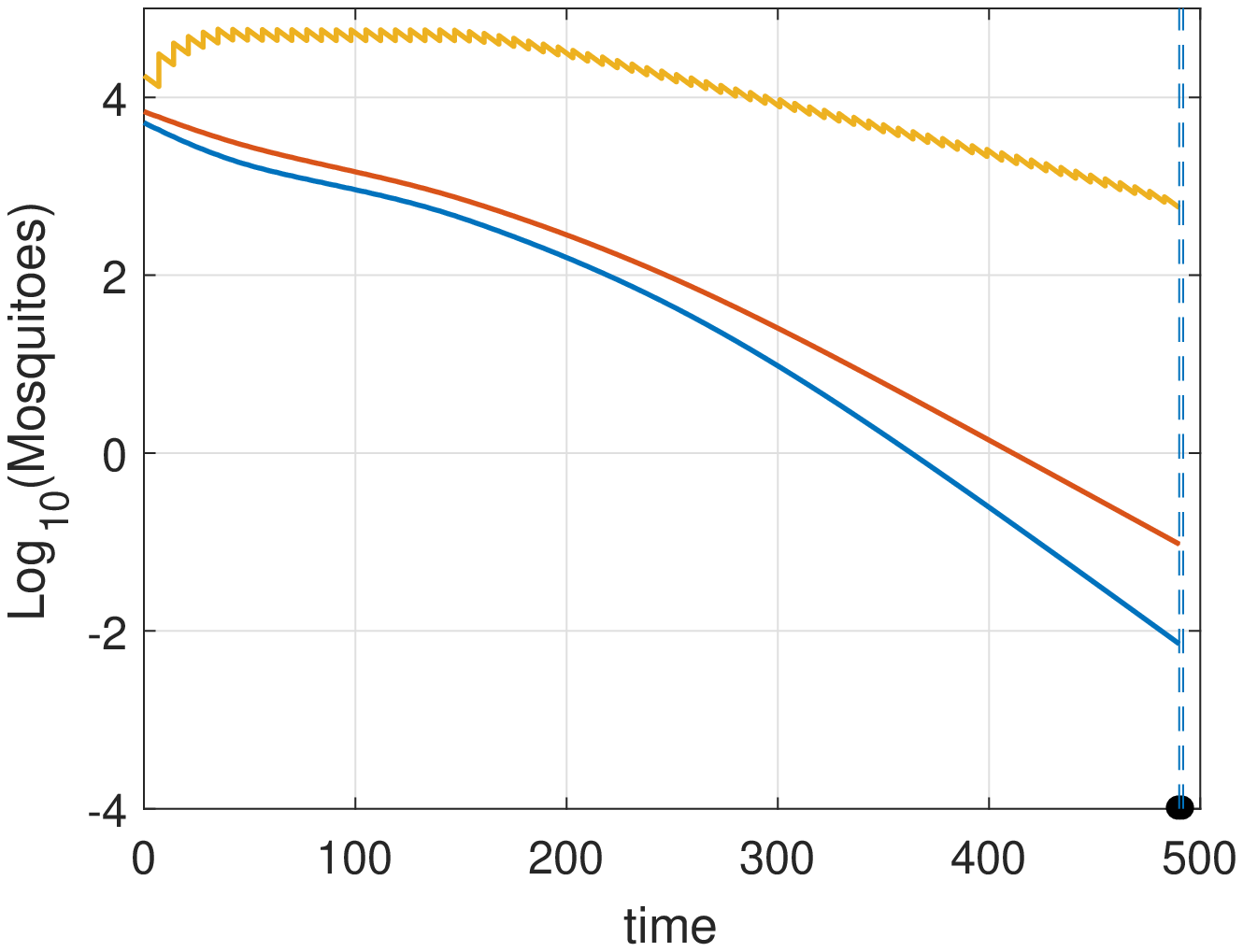} \\
 $\mathcal{P}_3=0, \; \hat{T}^{*}=504$ days; & $\mathcal{P}_3=10^3, \; \hat{T}^{*}=504$ days & $\mathcal{P}_3=10^5, \; \hat{T}^{*}=490$ days \\
$\text{CNSM}(U_n)= 434,820$ indv/ha & $\text{CNSM}(U_n)= 435,420$ indv/ha & $\text{CNSM}(U_n)= 457,500 $  indv/ha
\end{tabular}
\end{center}
\caption{\emph{Upper row:} suboptimal release programs $\{ U_n \}$ for $\mathcal{P}_3 \in \big\{ 0, 10^3, 10^5 \big\}$ with $\tau=7$ days. \emph{Lower row:} Time evolution of mosquito populations under respective $\{ U_n \}$: wild males $\log_{10} M(t)$ (blue-colored curves), wild females $\log_{10} F(t)$ (red-colored curves), and sterile males $\log_{10} M_S(t)$ (yellow-colored curves).  }
\label{fig2}
\end{figure}

As shown in the upper rows of Figures~\ref{fig2} and~\ref{fig3}, the release pick-values ($\tau u_{\max}$ per hectare) are known in advance, and the decision-makers may choose the frequency of releases ($\tau$ days) in accordance with the mass-rearing capacities of sterile insects available at situ.

\begin{figure}[t!]
\begin{center}
\begin{tabular}{ccc}
& suboptimal impulsive & \\
& release strategies & \\
\includegraphics[width=.33\textwidth]{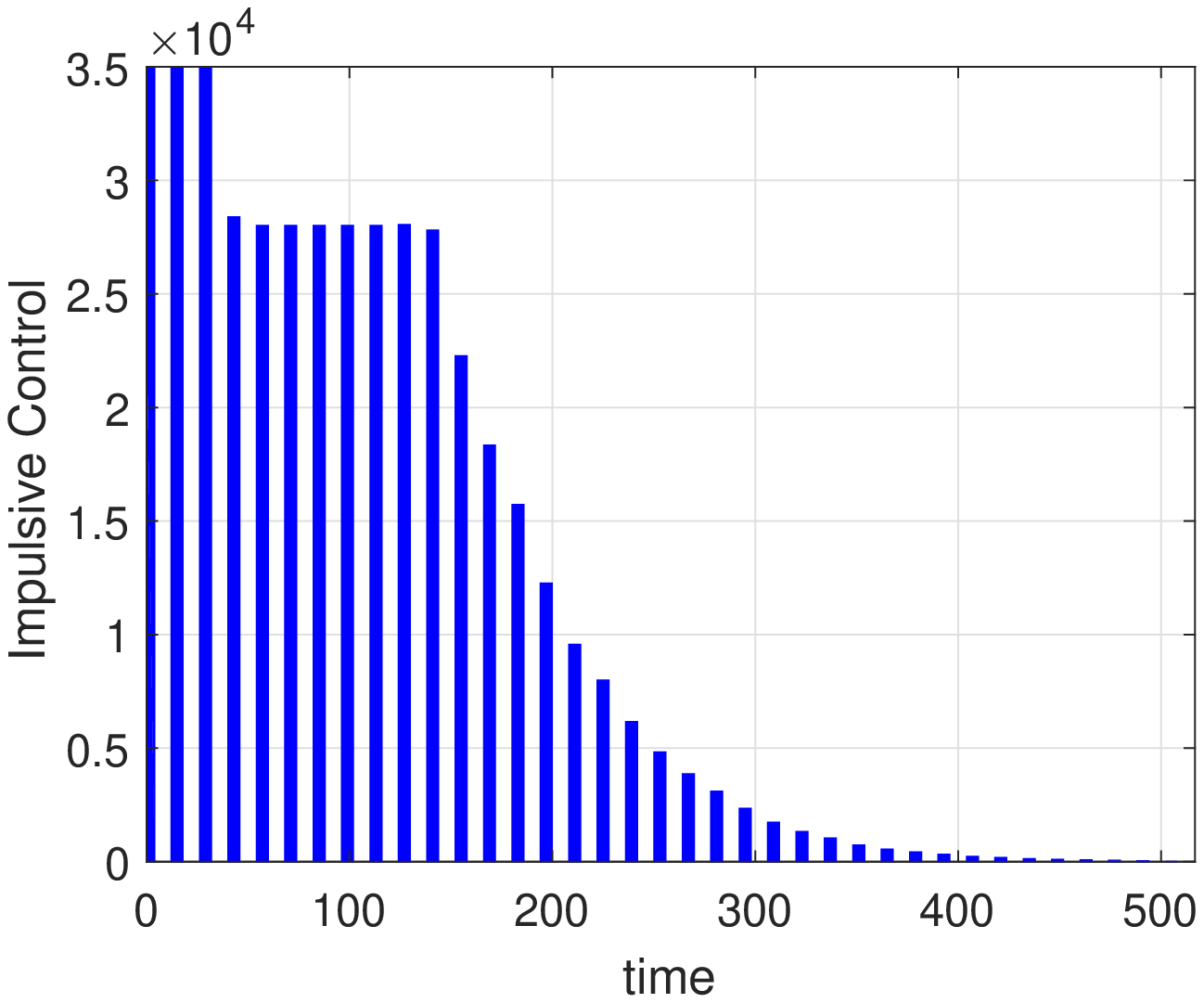} &
\includegraphics[width=.33\textwidth]{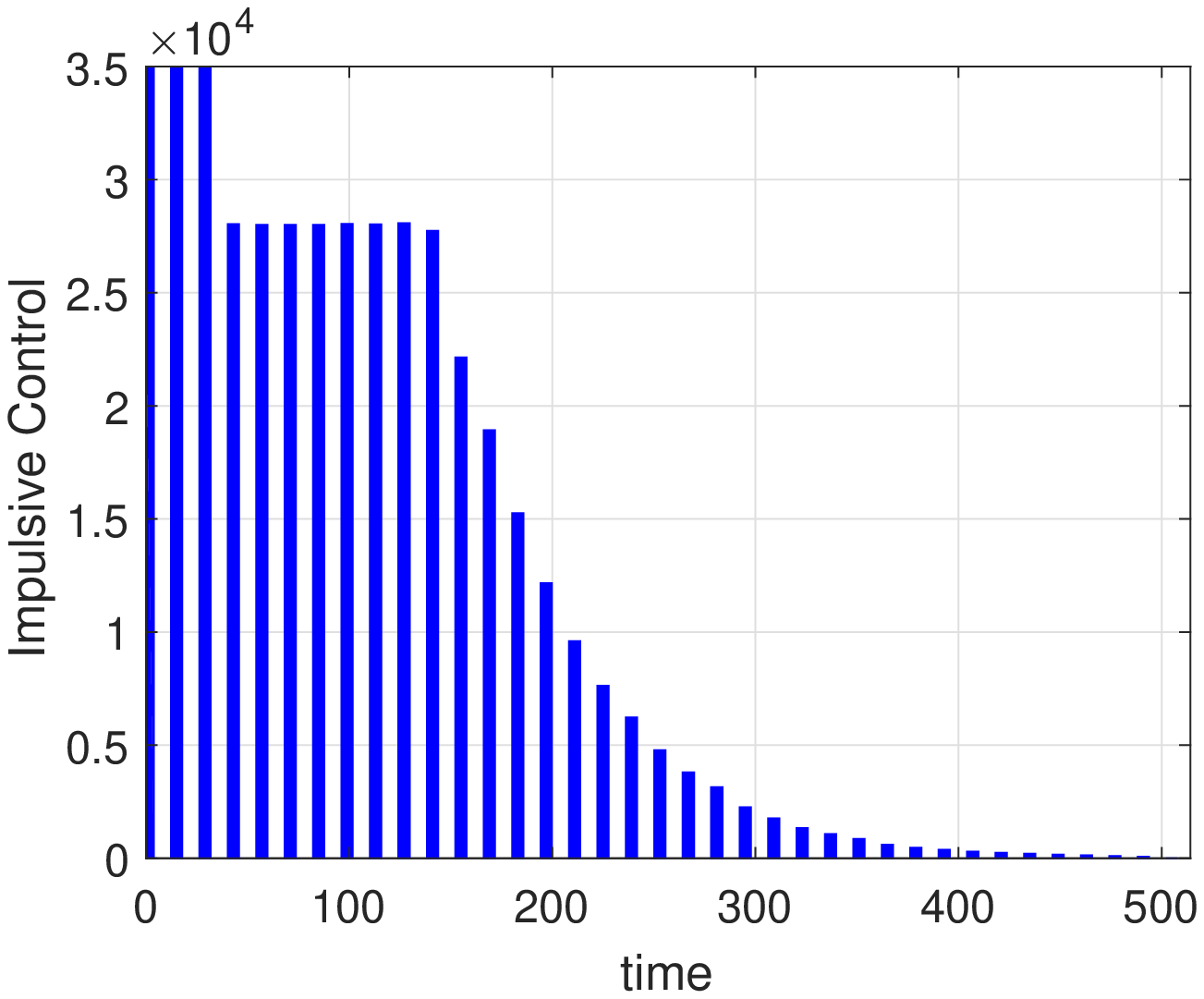} &  \includegraphics[width=.33\textwidth]{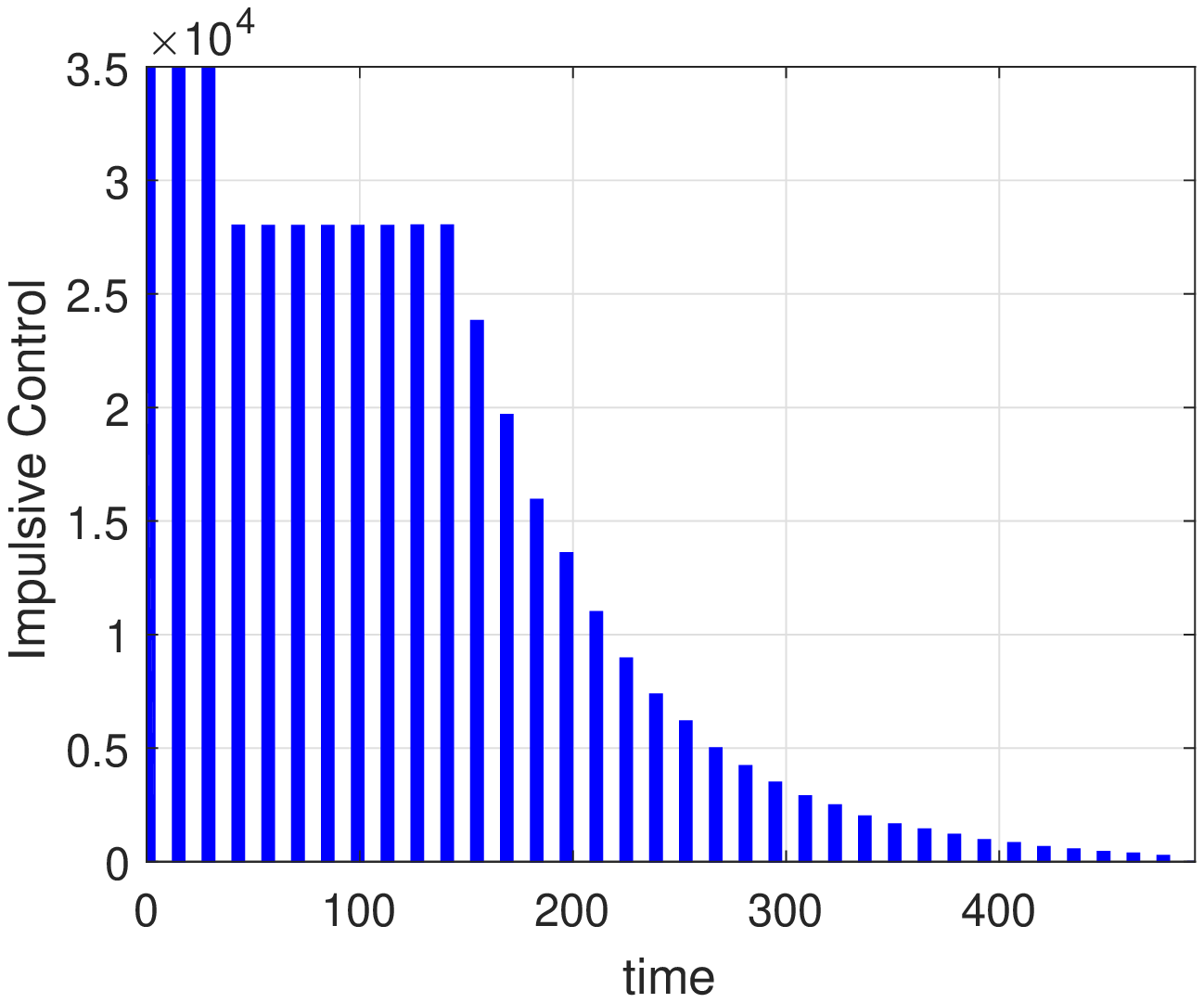} \\
& & \\
& Trajectories of suboptimal states & \\
& & \\
\includegraphics[width=.33\textwidth]{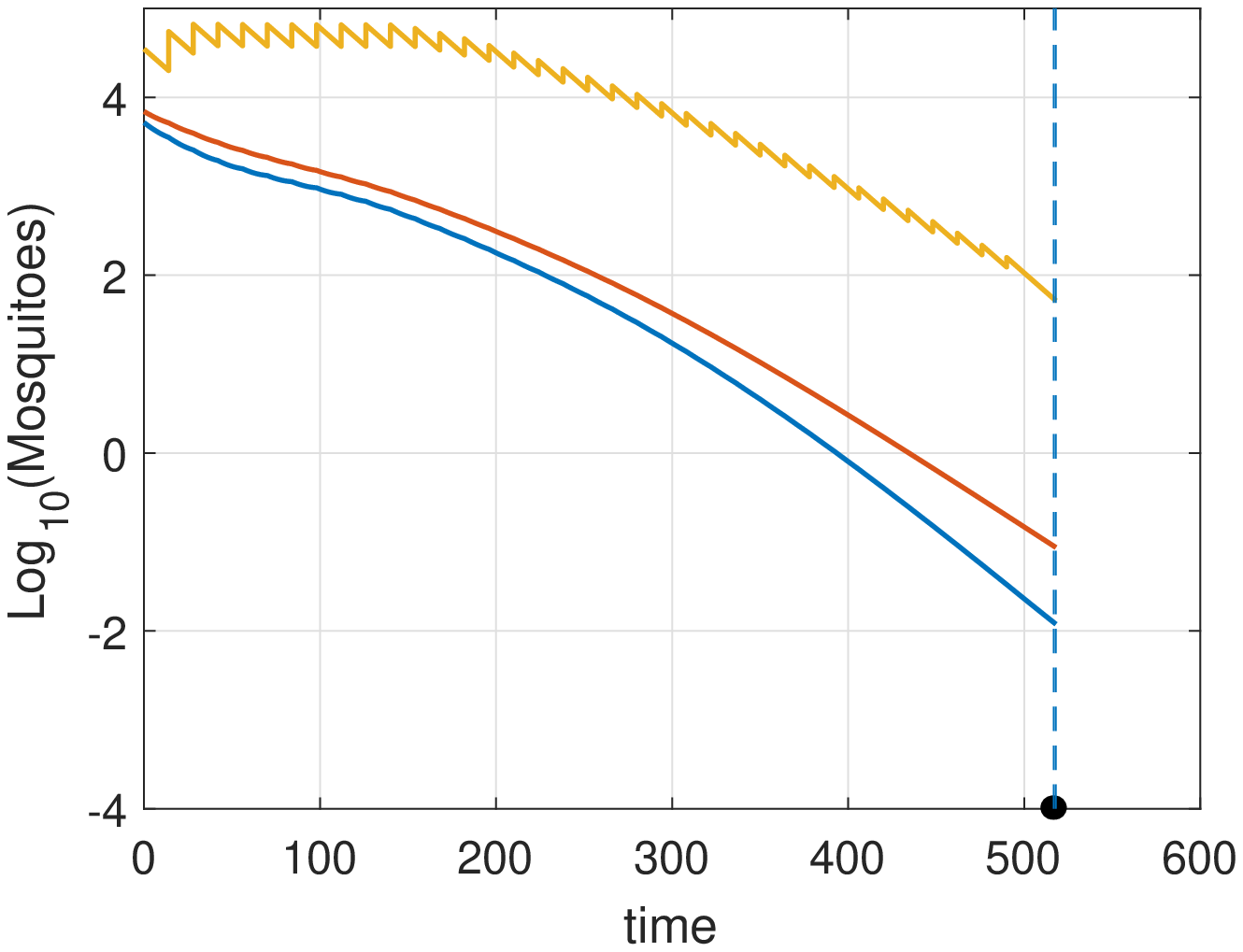} &  \includegraphics[width=.33\textwidth]{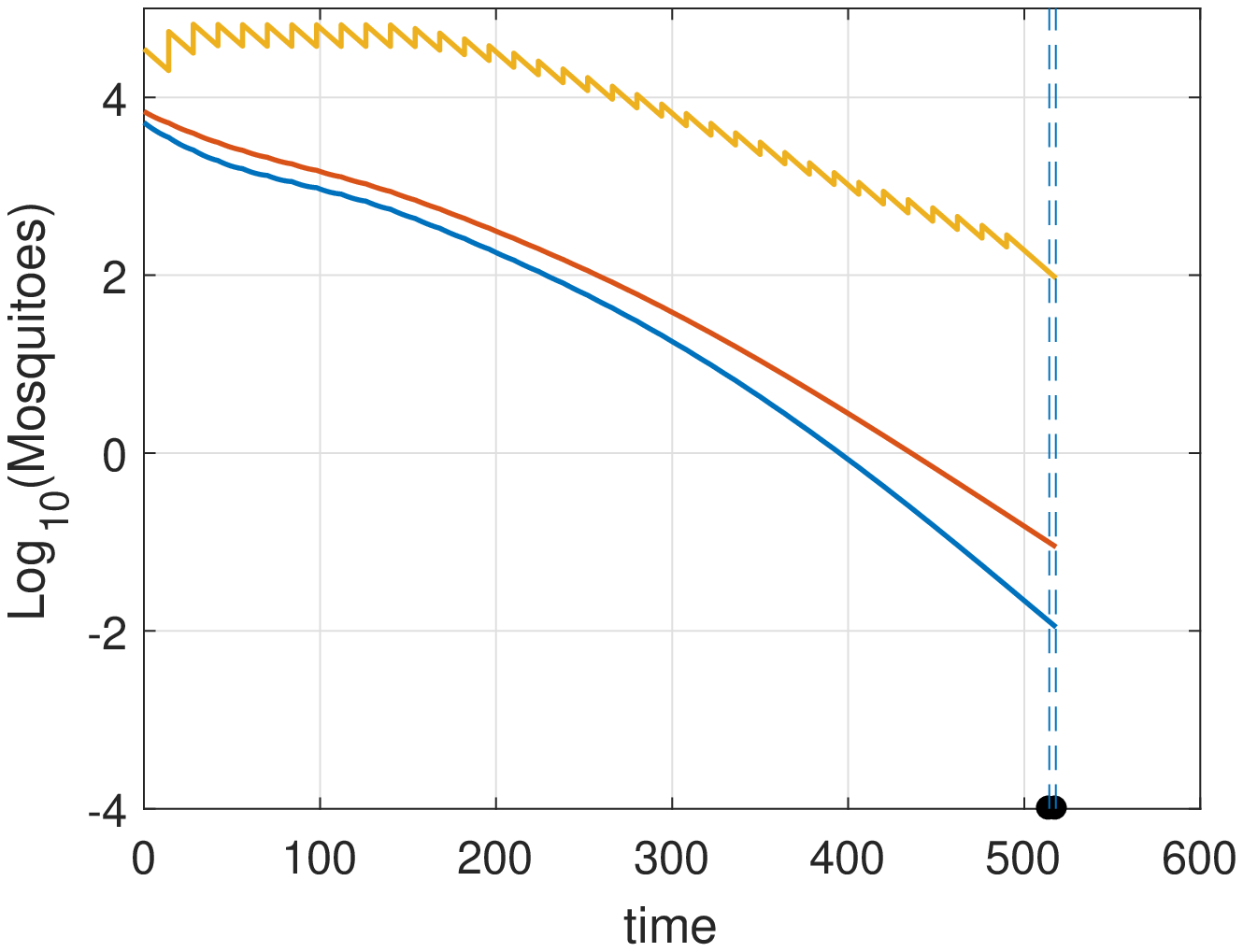}  & \includegraphics[width=.33\textwidth]{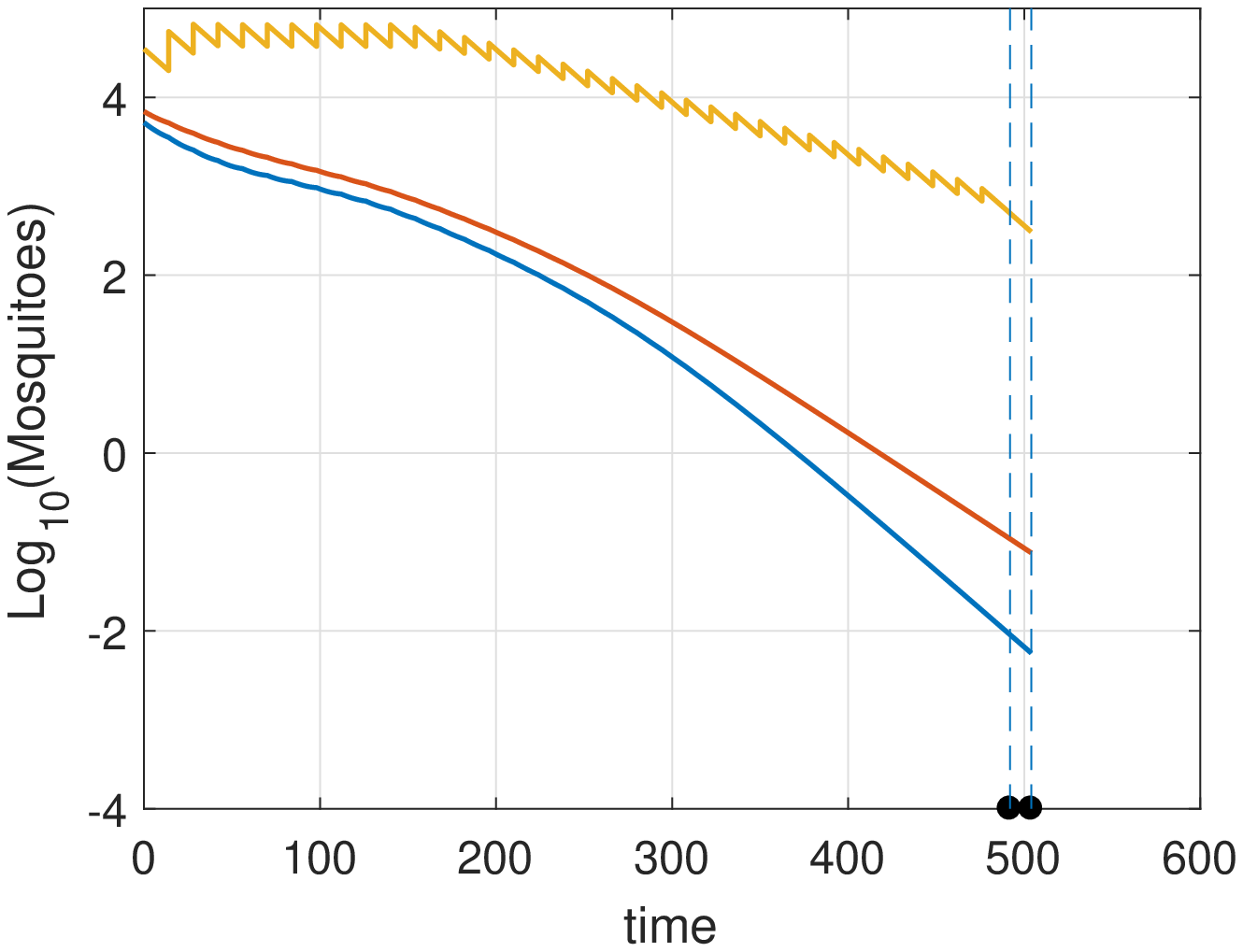} \\
 $\mathcal{P}_3=0, \; \hat{T}^{*}=518$ days; & $\mathcal{P}_3=10^3, \; \hat{T}^{*}=518$ days & $\mathcal{P}_3=10^5, \; \hat{T}^{*}=504$ days \\
$\text{CNSM}(U_n)= 442,480$ indv/ha & $\text{CNSM}(U_n)= 442,550$ indv/ha & $\text{CNSM}(U_n)= 464,050 $  indv/ha
\end{tabular}
\end{center}
\caption{\emph{Upper row:} suboptimal release programs $\{ U_n \}$ for $\mathcal{P}_3 \in \big\{ 0, 10^3, 10^5 \big\}$ with $\tau=14$ days. \emph{Lower row:} Time evolution of mosquito populations under respective $\{ U_n \}$: wild males $\log_{10} M(t)$ (blue-colored curves), wild females $\log_{10} F(t)$ (red-colored curves), and sterile males $\log_{10} M_S(t)$ (yellow-colored curves).  }
\label{fig3}
\end{figure}

In the following subsection, we compare our results displayed by Figures~\ref{fig2} and~\ref{fig3} for feedforward suboptimal impulsive SIT-control programs with mixed open/closed-loop impulsive SIT-control programs designed in \cite{Bliman2019}.

\subsection{Discussion of results}

It is worth noting that our simulations have been performed using exactly the same numerical values of the model's parameters (provided in Table~\ref{tab1}) as in the work authored by Bliman \emph{et al} \cite{Bliman2019}. Therefore, our results are comparable with those obtained in \cite{Bliman2019}.

It should be pointed out that suboptimal impulsive SIT-control programs designed for $\mathcal{P}_3=10^{3}$ (see middle columns in Figures~\ref{fig2} and~\ref{fig3}) do not make much sense. Namely, they require the same time $\hat{T}^{*}$ to reach elimination of the wild mosquitoes as those designed for $\mathcal{P}_3=0$ (see left columns in Figures~\ref{fig2} and~\ref{fig3}) while demanding to release more sterile insects. Therefore, these two suboptimal release programs are excluded from further consideration, and we focus on two key options with regards to priorities for decision-making, namely: (1) disregarding the time appreciation $\mathcal{P}_3=0$; (2) encouraging the time appreciation $\mathcal{P}_3=10^5$.

\begin{table}[t!]
  \centering
      \begin{tabular}{|p{22.5em}|c|c|c|}
    \toprule
    \multicolumn{1}{|c|}{\textbf{\large Impulsive release program}} & \multicolumn{1}{p{6.5em}|}{\textbf{Cumulative number of released sterile males}} & \multicolumn{1}{p{7.43em}|}{\textbf{Number of weeks needed to reach elimination}} & \multicolumn{1}{p{5.5em}|}{\textbf{Number of nonzero releases}} \\
    \midrule
    \multicolumn{1}{|r|}{} & \multicolumn{3}{c|}{Frequency of releases  = 7 days} \\
    \midrule
    Suboptimal without time appreciation ($\mathcal{P}_3 =0$) &       434.820    & 72    & 72 \\
    \midrule
    Suboptimal with time appreciation ($\mathcal{P}_3 =10^5$)  &       457.500    & 70    & 70 \\
    \midrule
    Mixed with synchronized measurements and larger control gain ($p=1, k=0.2/\cN_F$, \cite{Bliman2019}) &       450.668    & 72    & 72 \\
    \midrule
    Mixed with synchronized measurements and smaller control gain ($p=1, k=0.99/\cN_F$, \cite{Bliman2019}) &       457.489    & 246   & 246 \\
    \midrule
    Mixed with sparse measurements and larger control gain ($p=4, k=0.2/\cN_F$, \cite{Bliman2019}) &       534.849    & 65    & 53 \\
    \midrule
    Mixed with sparse measurements and smaller control gain ($p=4, k=0.99/\cN_F$, \cite{Bliman2019}) &       450.077    & 69    & 53 \\
    \midrule
    \multicolumn{1}{|r|}{} & \multicolumn{3}{c|}{Frequency of releases  = 14 days} \\
    \midrule
    Suboptimal without time appreciation ($\mathcal{P}_3 =0$)  & 442.480   & 74    & 37 \\
    \midrule
    Suboptimal with time appreciation ($\mathcal{P}_3 =10^5$)  &       464.050    & 72    & 36 \\
    \midrule
    Mixed with synchronized measurements and larger control gain ($p=1, k=0.2/\cN_F$, \cite{Bliman2019}) &       465.187    & 72    & 36 \\
    \midrule
    Mixed with synchronized measurements and smaller control gain ($p=1, k=0.99/\cN_F$, \cite{Bliman2019}) &       427.701    & 136   & 68 \\
    \midrule
    Mixed with sparse measurements and larger control gain ($p=4, k=0.2/\cN_F$, \cite{Bliman2019}) &       499.497    & 66    & 25 \\
    \midrule
    Mixed with sparse measurements and smaller control gain ($p=4, k=0.99/\cN_F$, \cite{Bliman2019}) &       449.099    & 74    & 28 \\
    \bottomrule
    \end{tabular}
  \label{tab2}
  \caption{Summary of simulation data for suboptimal (open-loop) and mixed (open/closed-loop) impulsive SIT-control programs}
\end{table}

Table~\ref{tab2} summarizes the key features of simulation data for suboptimal (open-loop) and mixed (open/closed-loop, \cite{Bliman2019}) impulsive SIT-control programs, such as the cumulative number of sterile males needed for successful implementation of the SIT-based campaign, the number of weeks to reach elimination of wild populations, and the number of effective (nonzero) releases to be performed during the whole SIT campaign.

It should be noted that the cumulative number of sterile males needed for successful implementation of an impulsive suboptimal SIT-control program is calculated as $\tau \sum \limits_{n=0}^{\infty} U_n$. In the first column of Table~\ref{tab2}, the control gain is defined as $\left( \dfrac{1}{k} -1 \right)$ and the parameter $k$ is chosen to satisfy $0 < k < \dfrac{1}{\cN_F}$. Thus, smaller values of $k$ implies a larger control gain, while its larger values imply a smaller control gain\footnote{The detailed explanation regarding the choice of $k$ is provided in \cite{Bliman2019}, where two values of $k$ are considered: $k=0.2/\cN_F$ and $k=0.2/\cN_F$ that correspond to the control gain of the order $\approx 378$ and $\approx 76$, respectively.}. Furthermore, $p=1$ stands for the measurements that are synchronized with the releases, while $p=4$ expresses the sparser measurements to be carried out every $4\tau$ days.

According to the second column of Table~\ref{tab2}, the suboptimal impulsive release programs require to mass-rear a lesser number of sterile mosquitoes than mixed open/closed-loop strategies with an exception of the one bearing $p=1, k=0.99/\cN_F$ (that is, a mixed strategy with synchronized measurements performed every $\tau=14$ days and a smaller control gain for the feedback mode). However, this release program needs the longest time to reach the elimination of wild mosquitoes besides numerous measurements of wild population sizes it requires to perform.

Analyzing the data from the third and fourth columns of Table~\ref{tab2}, it becomes clear that suboptimal release programs render about the same benefits as the mixed open/closed-loop programs based on the synchronized measurements while not requiring for assessments of wild population sizes and needing a lesser number of sterile insects. On the other hand, mixed open/closed-loop programs based on sparse measurements display better results than suboptimal open-loop programs in terms of the overall time needed to reach elimination and the number of effective (nonzero) releases. However, their advantages are counterpoised by the greater number of sterile insects needed for successful implementation of the SIT-control campaign and the extra costs for performing real-time assessments of wild population sizes.

In the end, the ultimate choice of the release program for SIT-control campaign must be made by practitioners after evaluating the realistic costs related to implementation of the SIT-control campaign, such as the costs for mass-rearing of larger/smaller cohorts of sterile insects, logistics costs for accomplishment of a single release, as well as the extra costs for taking a single real-time measurement of the wild population sizes.

\section{Conclusions}
\label{sec-con}

In this work, we have applied the dynamic optimization approach for design of feedforward (open-loop) continuous-time programs for SIT-control, which do not require to assess the sizes of wild populations in real time. We have also proposed their more realistic suboptimal variants for practical implementation in the field.

All designed programs have a very clear structure and exhibit monotonicity with respect to the quantity of sterile insects to be released at each day $t$ (in the case of continuous-time optimal programs) or at the commencement of each period of $\tau$ days (in the case of impulsive suboptimal programs). The fact that release sizes gradually decrease during the SIT-control campaign may actually help in the adequate planning of the infrastructure and underlying logistics of the SIT-control interventions.

Another advantage of open-loop impulsive release program is the anticipated knowledge of the release pick-values ($\tau u_{\max}$) that enables the proper choice of the release frequency $\tau$ in accordance with the mass-rearing capacities of sterile insects available at situ.

We have also explored the impact of the exogenous parameter $\mathcal{P}_3$ (expressing the time appreciation) on the structure of the designed release programs and their respective outcomes. From the latter, we have detected the existence of a certain tradeoff between the cumulative number of the released sterile insects needed to reach elimination and the overall duration of the SIT-control campaign. Namely, by increasing (significantly) the value of $\mathcal{P}_3$ (\emph{ceteris paribus}) the overall time $T^{*}$ (and also $\hat{T}^{*}$) can be (slightly) reduced on the cost of increasing the overall quantity of sterile males to be released during the SIT-based campaign (cf. left and right columns in Figures~\ref{fig1}-\ref{fig3}). However, the general structure of both optimal and suboptimal release programs exhibit robustness and resilience to variations in $\mathcal{P}_3$.

In this work, we have tried to keep the value of another exogenous parameter -- the daily mass-rearing capacity of sterile insects $u_{\max}$ -- at a moderate realistic level of $2,500$ individuals per hectare. However, a temperate enhancement (reduction) of $u_{\max}$ does not affect much the general structure of optimal and suboptimal release programs. The only difference we detected through numerical experiment (that are left beyond the scope of this paper) consists in less (more) initial releases to be performed at the maximum release capacity $u_{\max}$ per day (or $\tau u_{\max}$ every $\tau$ days).

Finally, the methodology presented in this paper can be easily adjusted to different biological characteristics of other mosquito or pest species (such as $\rho, \beta, \mu_M, \mu_F$) and of sterile insects ($\gamma, \mu_S$), while accounting for the total area of the target locality\footnote{It must be stressed again that all our simulations were carried out for a ``standardized'' area comprising 1 hectare.} and the underlying mass-rearing capacitates to produce a sufficient quantity of sterile males every $\tau$ days. Therefore, the outcomes of this study could virtually help in the design and implementation of successful SIT-control campaigns.

\section*{Acknowledgements}
Support from the Colciencias and ECOS-Nord Program (Colombia: Project CI-71089; France: Project C17M01) is kindly acknowledged. DC and OV were supported by the inter-institutional cooperation program MathAmsud (18-MATH-05, MOVECO project). YD is partially endorsed by the GEMDOTIS project, funded by the call ECOPHYTO 2018 (Action 27). YD is also (partially) supported by the DST/NRF SARChI Chair in Mathematical Models and Methods in Biosciences and Bioengineering at the University of Pretoria (grant 82770).

\bibliographystyle{plain}
\bibliography{BCDV_2019_bib}

\begin{thebibliography}{10}

\bibitem{Anguelov2012}
R.~Anguelov, Y.~Dumont, and J.~Lubuma.
\newblock Mathematical modeling of sterile insect technology for control of
  \emph{Anopheles} mosquito.
\newblock {\em Computers and Mathematics with Applications}, 64:374--389, 2012.

\bibitem{Bliman2019}
P.-A. Bliman, D.~Cardona-Salgado, Y.~Dumont, and O.~Vasilieva.
\newblock Implementation of control strategies for sterile insect techniques.
\newblock {\em Mathematical Biosciences}, 314:43--60, 2019.

\bibitem{Campo2018}
D.~Campo-Duarte, O.~Vasilieva, D.~Cardona-Salgado, and M.~Svinin.
\newblock Optimal control approach for establishing \textit{wMelPop Wolbachia}
  infection among wild \textit{Aedes aegypti} populations.
\newblock {\em Journal of Mathematical Biology}, 76(7):1907--1950, 2018.

\bibitem{Dumont2012}
Y.~Dumont and J.M. Tchuenche.
\newblock Mathematical studies on the sterile insect technique for the
  {C}hikungunya disease and \emph{Aedes albopictus}.
\newblock {\em Journal of mathematical Biology}, 65(5):809--854, 2012.

\bibitem{Dyck2006}
V.A. Dyck, J.~Hendrichs, and A.S. Robinson, editors.
\newblock {\em Sterile Insect Technique: Principles and Practice in Area-Wide
  Integrated Pest Management}.
\newblock Springer, Dordrecht, The Netherlands, 2005.

\bibitem{Fister2013}
K.~Fister, M.~McCarthy, S.~Oppenheimer, and C.~Collins.
\newblock Optimal control of insects through sterile insect release and habitat
  modification.
\newblock {\em Mathematical biosciences}, 244(2):201--212, 2013.

\bibitem{Fleming1975}
W.~Fleming and R.~Rishel.
\newblock {\em Deterministic and stochastic optimal control}, volume~1 of {\em
  Applications of Mathematics}.
\newblock Springer Verlag, New York, 1975.

\bibitem{Garg2011}
D.~Garg, M.~Patterson, C.~Francolin, C.~Darby, G.~Huntington, W.~Hager, and
  A.~Rao.
\newblock Direct trajectory optimization and costate estimation of
  finite-horizon and infinite-horizon optimal control problems using a {R}adau
  pseudospectral method.
\newblock {\em Computational Optimization and Applications}, 49(2):335--358,
  2011.

\bibitem{Huang2017}
M.~Huang, X.~Song, and J.~Li.
\newblock Modelling and analysis of impulsive releases of sterile mosquitoes.
\newblock {\em Journal of biological dynamics}, 11(1):147--171, 2017.

\bibitem{Lenhart2007}
S.~Lenhart and J.~Workman.
\newblock {\em Optimal Control Applied to Biological Models}.
\newblock CRC Press, Taylor \& Francis Group, Boca Raton, FL, 2007.

\bibitem{Li2015}
J.~Li and Z.~Yuan.
\newblock Modelling releases of sterile mosquitoes with different strategies.
\newblock {\em Journal of biological dynamics}, 9(1):1--14, 2015.

\bibitem{Liles1965}
J.~Liles.
\newblock Effects of mating or association of the sexes on longevity in
  \emph{Aedes aegypti} ({L}.).
\newblock {\em Mosquito News}, 25(4):434--439, 1965.

\bibitem{Manrakhan2015}
A.~Manrakhan, J.H. Venter, and V.~Hattingh.
\newblock The progressive invasion of \emph{Bactrocera dorsalis} ({D}iptera:
  {T}ephritidae) in {S}outh {A}frica.
\newblock {\em Biological Invasions}, 17(10):2803--2809, 2015.

\bibitem{Multerer2019}
L.~Multerer, T.~Smith, and N.~Chitnis.
\newblock Modeling the impact of sterile males on an \emph{Aedes aegypti}
  population with optimal control.
\newblock {\em Mathematical biosciences}, 311:91--102, 2019.

\bibitem{Oliva2012}
C.~Oliva, M.~Jacquet, J.~Gilles, G.~Lemperiere, P.-O. Maquart, S.~Quilici,
  F.~Schooneman, M.~Vreysen, and S.~Boyer.
\newblock The sterile insect technique for controlling populations of
  \emph{Aedes albopictus} ({D}iptera: {C}ulicidae) on {R}eunion {I}sland:
  mating vigour of sterilized males.
\newblock {\em PloS One}, 7(11):e49414, 2012.

\bibitem{Patterson2014}
M.~Patterson and A.~Rao.
\newblock {GPOPS-II: A MATLAB} software for solving multiple-phase optimal
  control problems using hp-adaptive {G}aussian quadrature collocation methods
  and sparse nonlinear programming.
\newblock {\em ACM Transactions on Mathematical Software (TOMS)}, 41(1):1,
  2014.

\bibitem{Seierstad1987}
A.~Seierstad and K.~Sydsaeter.
\newblock {\em Optimal control theory with economic applications}, volume~24 of
  {\em Advanced Textbooks in Economics}.
\newblock Elsevier North-Holland, Inc., Amsterdam, 1987.

\bibitem{Strugarek2019}
M.~Strugarek, H.~Bossin, and Y.~Dumont.
\newblock On the use of the sterile insect release technique to reduce or
  eliminate mosquito populations.
\newblock {\em Applied Mathematical Modelling}, 68:443--470, 2019.

\end{thebibliography}

\end{document}